\documentclass[american,aps,pra,reprint,superscriptaddress, longbibliography, nofootinbib]{revtex4-1}
\usepackage[T1]{fontenc}
\usepackage[latin9]{inputenc}
\usepackage{color}
\usepackage{babel}
\usepackage{amsthm}
\usepackage{amsmath}
\usepackage{amssymb}
\usepackage{graphicx}
\usepackage[unicode=true,pdfusetitle, bookmarks=true,bookmarksnumbered=false,bookmarksopen=false,  breaklinks=false,pdfborder={0 0 0},backref=false,colorlinks=false] {hyperref}
\hypersetup{ colorlinks,linkcolor=myurlcolor,citecolor=myurlcolor,urlcolor=myurlcolor}

\makeatletter
\@ifundefined{textcolor}{}
{%
 \definecolor{BLACK}{gray}{0}
 \definecolor{WHITE}{gray}{1}
 \definecolor{RED}{rgb}{1,0,0}
 \definecolor{GREEN}{rgb}{0,1,0}
 \definecolor{BLUE}{rgb}{0,0,1}
 \definecolor{CYAN}{cmyk}{1,0,0,0}
 \definecolor{MAGENTA}{cmyk}{0,1,0,0}
 \definecolor{YELLOW}{cmyk}{0,0,1,0}
 }
\theoremstyle{plain}
\newtheorem{thm}{\protect\theoremname}
\theoremstyle{plain}

\ifx\proof\undefined
\newenvironment{proof}[1][\protect\proofname]{\par
\normalfont\topsep6\p@\@plus6\p@\relax
\trivlist
\itemindent\parindent
\item[\hskip\labelsep
\scshape
#1]\ignorespaces
}{%
\endtrivlist\@endpefalse
}
\providecommand{\proofname}{Proof}
\fi
\theoremstyle{plain}
\newtheorem{lem}[thm]{\protect\lemmaname}
\newtheorem{defi}[thm]{\protect\definitionname}
\newtheorem{cor}[thm]{\protect\corollaryname}

\providecommand{\lemmaname}{Lemma}
\providecommand{\definitionname}{Definition}
\providecommand{\propositionname}{Proposition}
\providecommand{\corollaryname}{Corollary}
\providecommand{\theoremname}{Theorem}
\usepackage{babel}
\usepackage{txfonts}
\usepackage{colortbl}\definecolor{myurlcolor}{rgb}{0,0,0.7}
\newcommand{\tr}{{\operatorname{Tr\,}}}

\newcommand{\ketbra}[2]{|#1\rangle\!\langle#2|}
\def\ket#1{| #1 \rangle}
\def\bra#1{\langle  #1 |}
\def\braket#1{\langle  #1 \rangle}
\def\proj#1{| #1 \rangle\!\langle #1 |}

\newcommand{\haH}

\definecolor{orange}{RGB}{255,127,0}

\begin{document}
\title{Attaining Carnot Efficiency with Quantum and Nano-scale Heat Engines}

\author{Mohit Lal Bera}
\affiliation{ICFO -- Institut de Ci\`encies Fot\`oniques, The Barcelona Institute of Science and Technology, ES-08860 Castelldefels, Spain}

\author{Maciej Lewenstein} 
\affiliation{ICFO -- Institut de Ci\`encies Fot\`oniques, The Barcelona Institute of Science and Technology, ES-08860 Castelldefels, Spain}
\affiliation{ICREA, Pg.~Lluis Companys 23, ES-08010 Barcelona, Spain} 

\author{Manabendra Nath Bera}
\email{mnbera@gmail.com}
\affiliation{Department of Physical Sciences, Indian Institute of Science Education and Research (IISER), Mohali, Punjab 140306, India}

\begin{abstract}
A heat engine operating in the one-shot finite-size regime, where systems composed of a small number of quantum particles interact with hot and cold baths and are restricted to one-shot measurements, delivers fluctuating work. Further, engines with lesser fluctuation produce a lesser amount of deterministic work. Hence, the heat-to-work conversion efficiency stays well below the Carnot efficiency. Here we overcome this limitation and attain Carnot efficiency in the one-shot finite-size regime, where the engines allow the working systems to simultaneously interact with two baths via the semi-local thermal operations and reversibly operate in a one-step cycle. These engines are superior to the ones considered earlier in work extraction efficiency, and, even, are capable of converting heat into work by exclusively utilizing inter-system correlations. We formulate a resource theory for quantum heat engines to prove the results.
\end{abstract}

\maketitle

\section{Introduction}
Heat engines are the fundamental building blocks of modern technology. These were invented primarily to convert heat into mechanical work. To lay a theoretical framework and to uncover the laws governing the processes in the engines, the thermodynamics was empirically developed \cite{SadiCarnot}. Later, it has been founded on statistical mechanics \cite{Callen85} assuming that the systems are large and composed of an asymptotically large number of particles ($N \to \infty$) interacting with even larger baths, where the average fluctuation in energy approaches zero. This is termed usually as the asymptotic regime.  

However, the situation changes completely for the systems of a finite, but moderate or even small number of quantum particles ($N \ll \infty$) where the standard thermodynamics is not applicable. In such cases, from the very beginning, the fluctuations may play a much more important role. The situations may be classified in two regimes: many-shot finite-size regime - where repeated measurements (in time) are allowed on a system is made up of moderate or a small number of particles, and one-shot finite-size regime - where only one-shot measurements are allowed on a system composed of a single or a moderate number of particles. In the last decades, enormous efforts have been put forward to extend thermodynamics to these regimes leading to two major approaches. 

The first approach is based on fluctuation theorems, exploiting statistical mechanics and open quantum systems dynamics \cite{Jarzynski97, Crooks99, Campisi11}. The other one is based on the quantum information theory \cite{Brandao13, Skrzypczyk14, Horodecki13, Aberg13, Brandao15, Lostaglio15, Cwiklinski15, Lostaglio15a, Bera16, Bera17, Sparaciari17, Sparaciari17, Bera17, Uzdin18, Muller18, Gour2018}. Among others, the latter leads to a resource theory of quantum systems out of thermal equilibrium, which is commonly termed as the resource theory of quantum thermodynamics (RTQTh) \cite{Brandao13, Horodecki13, Brandao15}. The RTQTh stands out among the other approaches as it exploits a rigorous mathematical framework similar to the resource theory of entanglement. Recently, the approaches based on fluctuation theory and resource theory have been inter-connected for some cases \cite{Alhambra16, Aberg18, Guarnieri19}. All these investigations are majorly limited to the situations, where the quantum system is interacting with only one thermal bath at a fixed temperature. Apart from some efforts to quantify extractable work and engine efficiency in few special cases, and to study the finite-size effects and the quantum signatures \cite{Alicki79, Kieu04, Verley14, Rosnagel14, Uzdin15, Tajima17, Ng17, Ito18, Woods2019, Woods19a}, there has been no major progress, so far, in formulating a resource theory for quantum heat engines. 

One of the striking features is that the engines operating in the one-shot finite-size regime can only deliver fluctuating work \cite{Horodecki13, Brandao15} and the lesser the fluctuation lesser becomes the extractable deterministic or one-shot work.  Further, a heat engine operating in this regime cannot in general achieve reversible transformation. As a consequence, it is not possible to attain the maximum allowed heat-to-work conversion efficiency (i.e., the Carnot efficiency) in a Carnot engine, unless the system interacting with the baths is made up of an asymptotically large number of particles where Carnot efficiency may at most be achieved on average.   

In this work, we present quantum and nano-scale heat engines that attain the maximum possible heat-to-work conversion efficiency, i.e., the Carnot efficiency, in the one-shot finite-size regime. These engines are superior in work extraction efficiency compared to the traditional engines. To prove our results and to address quantum thermodynamics in the one-shot finite-size regime in general, we formulate a resource theory for quantum heat engines (see Supplementary Information) in which a system with few quantum particles interacts with two or multiple thermal baths. With the precise characterization of thermodynamic operations by introducing a first law for engines, we derive the second laws for quantum state transformation in presence two or multiple baths at different temperatures by using information-theoretic tools. The newly introduced engine operations are more general in the sense that the system interacts with the baths simultaneously. We term these engine operations as ``semi-local thermal operations'' (SLTOs). The SLTOs not only enable us to build a Carnot heat engine operating with a one-step cycle, but also enhances the work extraction efficiency in the one-shot finite-size regime -- in this sense the SLTOs are more powerful than the ones considered earlier. As revealed by this resource theoretic framework, the state transformations in the quantum engine are fundamentally irreversible and must obey many second laws. As an important result of this framework, we design a reversible engine transformation that attains the maximum possible efficiency for work extraction, i.e., the Carnot efficiency, there by demonstrating the supremacy of quantum heat engines.

\begin{figure}
\includegraphics[width=0.8\columnwidth]{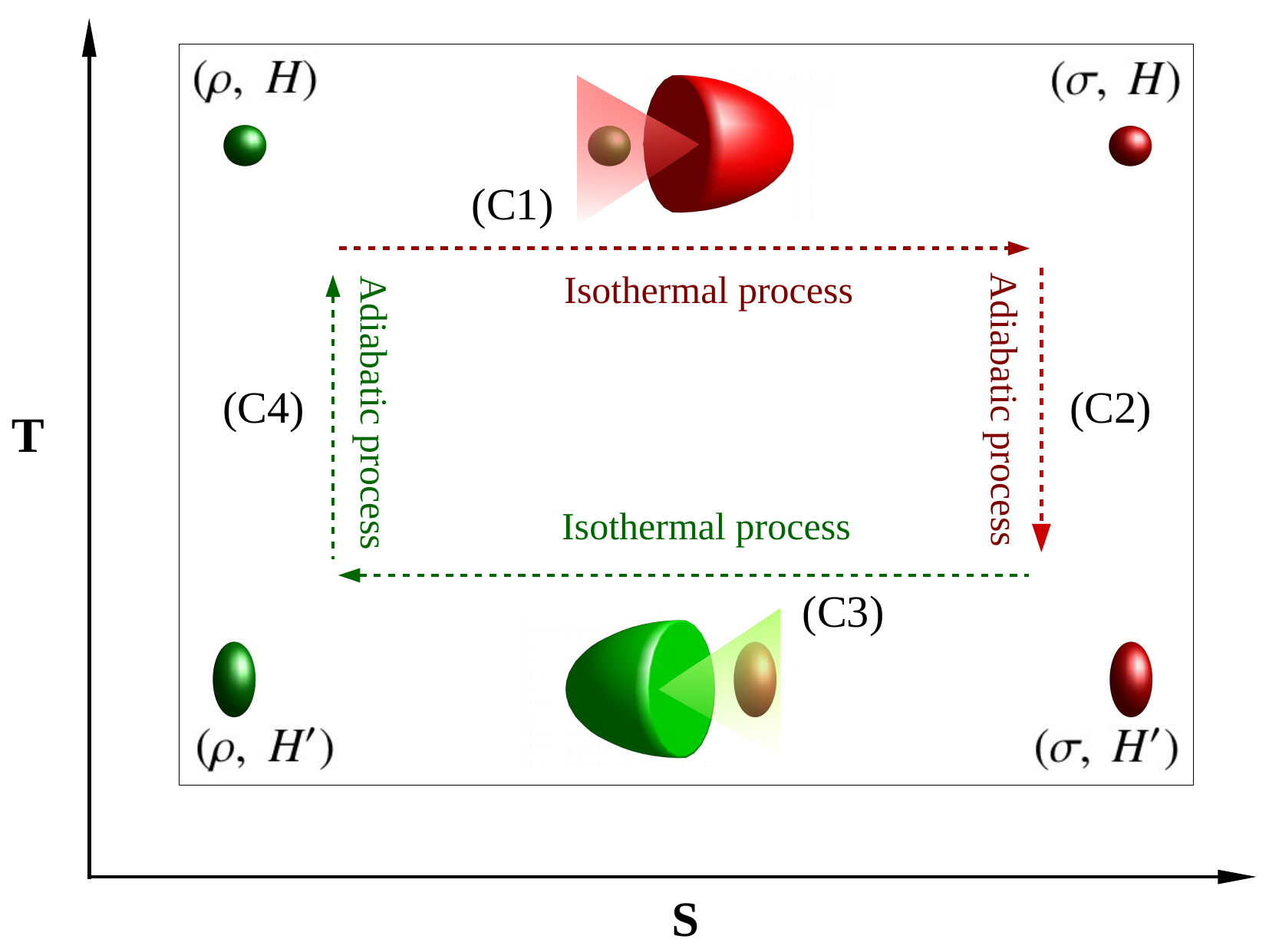}
\caption{{\bf A schematic of the operations in a traditional Carnot heat engine.} The horizontal and vertical axes are the thermodynamic entropy ($S$) and the temperature ($T$). The engine is made up of one working system and two heat baths with inverse temperatures $\beta_1=1/T_1$ and $\beta_2=1/T_2$, where $\beta_1 < \beta_2$. The engine operates in a cycle composed of four thermodynamically reversible steps: (C1) First, an isothermal transformation $(\rho, H) \rightarrow (\sigma, H) $ in interaction with the bath $B_1$ at inverse temperature $\beta_1$, where the state changes ($\rho \rightarrow \sigma$) without updating the system Hamiltonian $H$. (C2) Second, an adiabatic transformation $(\sigma, H) \rightarrow (\sigma, H^\prime)$ without any contact with the baths, where state remain unchanged but the system Hamiltonian modifies to $H \rightarrow H^\prime$. (C3) Third, an isothermal transformation $(\sigma, H^\prime) \rightarrow (\rho, H^\prime)$ in interaction with the bath $B_2$ at inverse temperature $\beta_2$, only changing the state. (C4) Finally, an adiabatic transformation $(\rho, H^\prime) \rightarrow (\rho, H)$ without any interaction with the baths, and updating only the system Hamiltonian $H^\prime \rightarrow  H$. The steps (C1)-(C4) constitute a cycle and the engine operates repeating the cycle many times. The important point is that, in the steps (C1) and (C2), the working system interacts with one bath at a time.
\label{fig:QHEng}}
\end{figure}

\section{Results}

A typical (Carnot) heat engine is comprised of two heat baths $B_1$ and $B_2$ with the inverse temperatures $\beta_1=1/T_1$ and $\beta_2=1/T_2$ respectively and a working system $S$, as shown in the Figure  \ref{fig:QHEng}. We assume $\beta_1 < \beta_2$ throughout this article. In our proposed quantum heat engine, we define a new engine operation where a working system ($S_{12}$) is composed of two non-interacting subsystems $S_{1}$ and $S_2$. The Hamiltonian is $H_{S_{12}}=H_{S_1}+H_{S_2} \equiv H_{S_1}\otimes \mathbb{I}_{S_2} + \mathbb{I}_{S_1} \otimes H_{S_2}$. The Hamiltonians of the baths $B_1$ and $B_2$ are denoted as $H_{B_{1}}$ and $H_{B_2}$ respectively. The subsystems $S_1$ and $S_2$ semi-locally interact with the baths $B_1$ and $B_2$ respectively. With this, the number of steps in the Carnot engine is reduced. For instance, consider that the subsystems $S_1$ and $S_2$ are in the states $\rho$ and $\sigma$ respectively. Then the isothermal steps (C1) and (C3) can be combined to one step, as 
\begin{align*}
 (\rho \otimes \sigma, \ H_{S_1} + H_{S_2}) \rightarrow (\sigma \otimes \rho, \ H_{S_1} + H_{S_2}),
\end{align*}
where the $H_{S_1}=H$ and the $H_{S_2}=H^\prime$. In this step, the subsystems swap their states without changing their Hamiltonian. Further, both the adiabatic steps (C2) and (C4) can be performed in one step as well, that is
\begin{align*}
 (\sigma \otimes \rho, \ H_{S_1} + H_{S_2}) \rightarrow (\sigma \otimes \rho, \ H_{S_1}^\prime + H_{S_2}^\prime),
\end{align*}
where the $H_{S_1}^\prime=H^\prime$ and $H_{S_2}^\prime=H$. Here the subsystems swap their local Hamiltonians without modifying their states. In fact, the four steps in a Carnot heat engine can be further reduced to just one step (see Figure \ref{fig:OurEng}) which enables one to attain maximum possible heat-to-work conversion efficiency, as we shall discuss later. 

\subsection{Semi-local thermal operations}

Let us now introduce the general form of thermodynamically allowed (semi-local) operations that a (bipartite) quantum system $S_{12}$ undergoes in a quantum heat engine, where the bipartite system $S_{12}$ can be in an arbitrary state. Even, the states may possess strong correlation, e.g., quantum entanglement, shared by the subsystems $S_1$ and $S_2$. 

\begin{defi}[Semi-local thermal operations (SLTOs)] \label{propmt:slto}
In a quantum heat engine, the thermodynamic operations on system $S_{12}$ in a state $\rho_{S_{12}}$ are defined as 
\begin{align}
\label{eqmt:slto-Stinespring}
 \Lambda_{S_{12}} \left(\rho_{S_{12}}\right)=\tr_{B_1 B_2} \left[ U (\gamma_{B_1} \otimes \gamma_{B_2} \otimes \rho_{S_{12}}) U^\dag \right],
\end{align} 
with the condition that the global unitary $U$ satisfies the commutation relations
\begin{align}
\left[U, \ H_{B_1} + H_{S_1} +  H_{B_2} + H_{S_2}  \right] &=0, \label{eqmt:slto-commutation} \\
\left[U, \ \beta_1 \ (H_{B_1} + H_{S_1}) + \beta_2 \ (H_{B_2} + H_{S_2})  \right] &=0, \label{eqmt:slto-commutationTemp}
\end{align}
where the thermal states of the baths are denoted by $\gamma_{B_x}=\frac{e^{-\beta_xH_{B_x}}}{\tr [e^{-\beta_xH_{B_x}}]}$ for $x=1,2$.
\end{defi}
The resultant operations on the system $S_{12}$ are semi-local in the sense that, even though the subsystems ($S_1$ and $S_2$) ``selectively'' interact with the baths ($B_1$ and $B_2$), the unitary $U$ still allows certain interactions among them with the constraints \eqref{eqmt:slto-commutation} and \eqref{eqmt:slto-commutationTemp}. It should be noted that the commutation relations \eqref{eqmt:slto-commutation} and \eqref{eqmt:slto-commutationTemp} together constitute the first law for quantum heat engines. The relation \eqref{eqmt:slto-commutation} guarantees strict conservation of the total energy $E_{12}=E_1 + E_2$, where $E_1$ and $E_2$ are the energies of the $B_1S_1$ and $B_2S_2$ composites respectively. In addition, the relation \eqref{eqmt:slto-commutationTemp} ensures strict conservation of the total weighted-energy $E_{12}^{\beta_1\beta_2}=\beta_1 E_1 + \beta_2 E_2$, and it signifies that any change in (one-shot) entropy of $B_{1}S_1$, due to an exchange of energy between $B_1S_1$ and $B_2S_2$, must be compensated by a counter change in (one-shot) entropy of $B_2S_2$. It is interesting to note that the SLTOs converge to the (local) thermal operations that are introduced in the resource theory of quantum states beyond thermal equilibrium presented in \cite{Brandao13, Horodecki13, Brandao15}, when both the baths are of the same temperature, i.e., for $\beta_1=\beta_2$. Several useful properties of SLTOs are outlined in the Method. 

The SLTOs can be further generalized with an access to a bipartite catalyst $C_{12}$ composed of two non-interacting subsystems $C_1$ and $C_2$ and the Hamiltonian $H_{C_{12}}=H_{C_1} + H_{C_2}$. The $C_1$ is clubbed with the subsystem $S_1$ to form the composite $S_1C_1$. Similarly, the $C_2$ is clubbed with the $S_2$ to form $S_2C_2$. Then, the composites $S_1C_1$ and $S_2C_2$ interacts with the baths $B_1$ and $B_2$ via semi-local thermal operations.  Such operations are called catalytic semi-local thermal operations (cSLTOs) that satisfy 
\begin{align}
 \Lambda_{S_{12}C_{12}}(\rho_{S_{12}} \otimes\rho_{C_{12}} ) \rightarrow \sigma_{S_{12}} \otimes \rho_{C_{12}},
\end{align} 
where $\rho_{C_{12}} $ is a state of the catalyst. Note, the catalyst remains unchanged before and after the process. These catalytic operations form a larger set of thermodynamically allowed operations compared to SLTOs and respect all the properties satisfied by the SLTOs. The cSLTOs are the allowed thermodynamic operation in a quantum heat engine and constitute the free operation for the resource theory developed to prove the results presented in this article (see Supplementary Information).

When the subsystems are locally in thermal equilibrium with the baths they are semi-locally interacting with, the joint uncorrelated state of the system $S_{12}$ becomes $\gamma_{S_{12}}=\gamma_{S_1} \otimes \gamma_{S_2}$,  where $\gamma_{S_x}=e^{-\beta_xH_{S_x}}/{Z_x} $ with the partition functions $Z_{x}=\tr [e^{-\beta_xH_{S_x}}]$ for $x=1,2$. We term these states as the semi-Gibbs states. The set of all such semi-Gibbs states is denoted by the set $\mathcal{T}_{S_{12}} \ni \gamma_{S_{12}}$. The SLTOs map the set $\mathcal{T}_{S_{12}}$ onto itself. The SLTOs and the semi-Gibbs states are the precursors of a resource theory of heat engines that we develop in the Supplementary Information.

\begin{figure}
\includegraphics[width=0.8\columnwidth]{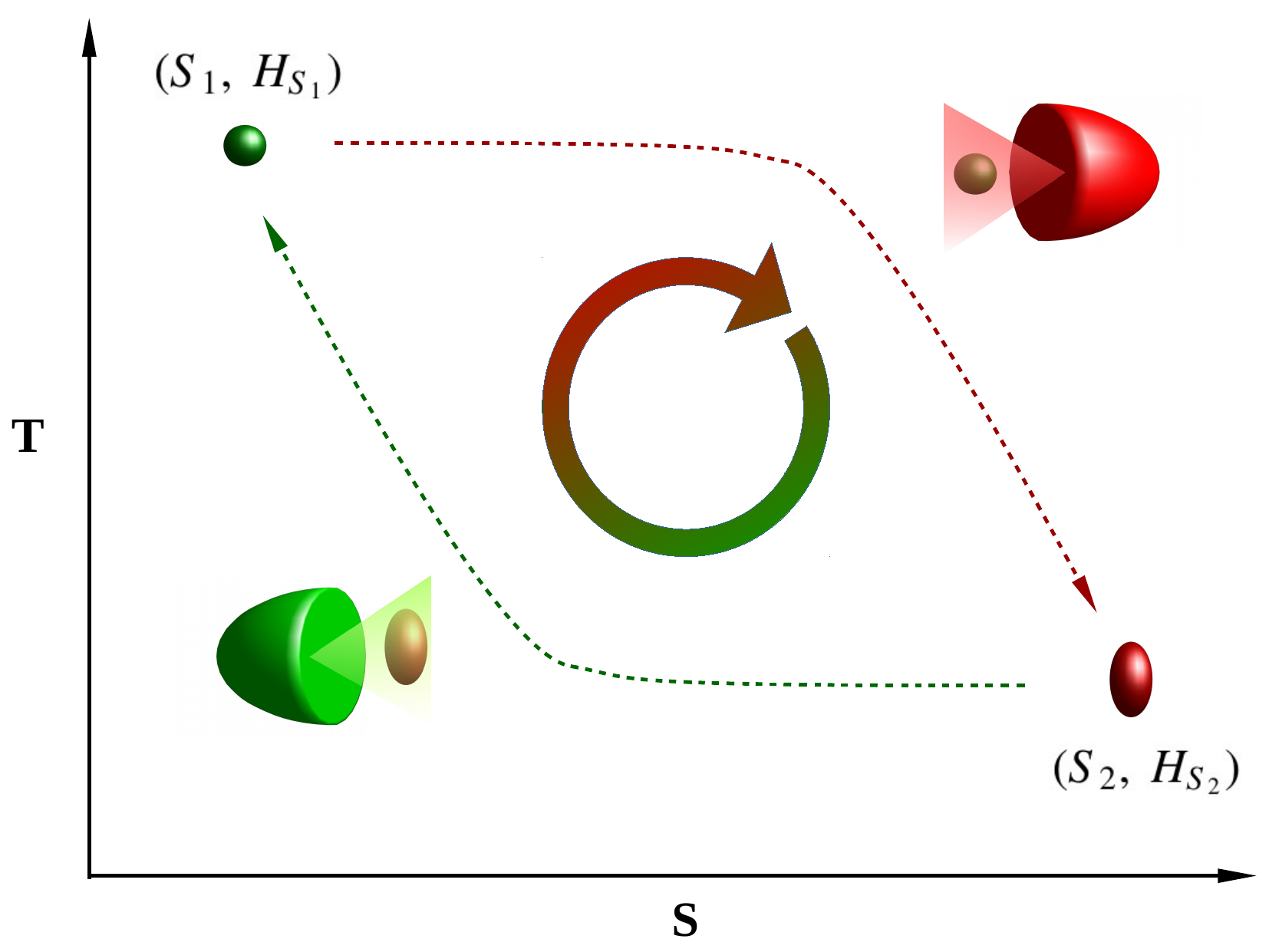}
\caption{{\bf Generalized one-step engine operation.} Consider initial state and the non-interacting Hamiltonian of the working-system are $\rho_{S_{12}}$ and $H_{S_{12}}=H_{S_1} + H_{S_2}$ respectively. The sub-systems $S_1$ and $S_2$ semi-locally interact with the baths $B_1$ and $B_2$ with inverse temperatures $\beta_1$ and $\beta_2$ respectively, where $\beta_1 < \beta_2$. The engine operates in a cycle by implementing the step $\left(\rho_{S_{12}}, \ H_{S_{12}} \right) \longrightarrow \left( \sigma_{S_{12}}^\prime, \ H_{S_{12}}^\prime  \right)$, with the modified system Hamiltonian $H_{S_{12}}^\prime=H_{S_1}^\prime + H_{S_2}^\prime$, so that it satisfies the conditions $\sigma_{S_{12}}^\prime=U^{SWAP}_{S_1 \leftrightarrow S_2} \left( \rho_{S_{12}} \right), \ \ H_{S_1}^\prime=H_{S_2}, \ \ \mbox{and} \ \ H_{S_2}^\prime=H_{S_1}$, where the unitary $U^{SWAP}_{S_1 \leftrightarrow S_2}$ performs a SWAP operation between sub-systems $S_1$ and $S_2$. Note,  all four steps in the traditional Carnot engine (comprising the steps (C1)-(C4) in Figure \ref{fig:QHEng}) can be performed in one-stroke with $\rho_{S_{12}}=\rho \otimes \sigma$, which is just a special case of above general operation.
\label{fig:OurEng}}
\end{figure}  

\subsection{Engine with one-step cycle, second laws, and work extraction}

The first important result of this work is that the SLTOs can be utilized to implement the cycle in a Carnot engine in one-step, as shown in Figure \ref{fig:OurEng}. At the end of a cycle, the initial state of the system is returned. Such a one-step engine cycle involves the transformation
\begin{align}\label{eq:EngineTransformGen}
 \left(\rho_{S_{12}}, \ H_{S_{12}} \right) \longrightarrow \left( \sigma_{S_{12}}^\prime, \ H_{S_{12}}^\prime  \right).
\end{align} 
Here the system Hamiltonian is modified to $H_{S_{12}}=H_{S_1} + H_{S_2} \to H_{S_{12}}^\prime=H_{S_1}^\prime + H_{S_2}^\prime$, and satisfies the conditions  $\sigma_{S_{12}}^\prime=U^{SWAP}_{S_1 \leftrightarrow S_2} \left( \rho_{S_{12}} \right), \ \ H_{S_1}^\prime=H_{S_2}, \ \ \mbox{and} \ \ H_{S_2}^\prime=H_{S_1}$, where the unitary $U^{SWAP}_{S_1 \leftrightarrow S_2}$ swaps the subsystems $S_1$ and $S_2$. The next cycle starts by inter-changing the interactions between subsystems and baths. In other words, the first engine cycle starts with the (semi-local) interactions as $B_1S_1-S_2B_2$, where the subsystems $S_1$ and $S_2$ semi-locally interact with the baths $B_1$ and $B_2$ respectively. In the next cycle, engine interchanges the interactions as $B_1S_2-S_1B_2$, where the subsystem $S_1$ and $S_2$ semi-locally interact with the baths $B_2$ and $B_1$ respectively via a cSLTO, and the cycles continue.

Here we restrict ourselves to the states $\rho_{S_{12}}$ that are block-diagonal in the energy eigenbases of the system Hamiltonian $H_{S_{12}}$, i.e., $[\rho_{S_{12}}, H_{S_{12}}]=0$. The cSLTOs are time-translation symmetric with respect to the time-translation driven by $H_{S_{12}}$, and that is why the cSLTOs monotonically decrease the superpositions between different energy eigenbases (see Appendix). For a heat engine operates in an arbitrarily large number of cycles, it is therefore safe to assume that an arbitrary state will dephase to its block-diagonal form after some cycles. Then the second laws that provide the necessary and sufficient conditions for such transformations are given in the theorem below. This theorem is proven, even for more general transformation, in the Supplementary Information.

\begin{thm}[Second laws for engines] \label{thmmt:2ndLawsAll}
Under cSLTOs, the transformation in Eq. \eqref{eq:EngineTransformGen}
is possible if, and only if, 
\begin{align}\label{eqmt:2ndLawsAll}
S_{\alpha}\left(\rho_{S_{12}}, \gamma_{S_{1}} \otimes \gamma_{S_{2}} \right) \geqslant S_{\alpha}\left(\sigma_{S_{12}}^\prime, \gamma_{S_{1}}^\prime \otimes \gamma_{S_{2}}^\prime \right), \ \ \forall \alpha \geqslant0,
\end{align}
where the $\alpha$-free-entropy of $\rho_{S_{12}}$ (and similarly for $\sigma_{S_{12}}^\prime$) is defined, for all $\alpha \in [-\infty, \infty]$, as
\begin{align}\label{eqmt:free-ent}
S_\alpha (\rho_{S_{12}}, \gamma_{S_1}\otimes \gamma_{S_2})= D_\alpha \left(\rho_{S_{12}} \parallel \gamma_{S_1}\otimes \gamma_{S_2} \right) - \log Z_1Z_2,  
\end{align}
with the thermal states $\gamma_{S_x}=\frac{e^{-\beta_x H_{S_x}}}{Z_x}$ and $\gamma_{S_x}^\prime=\frac{e^{-\beta_x H_{S_x}^\prime}}{Z_x^\prime}$, and the partition functions $ Z_x=\tr[e^{-\beta_x H_{S_x}}]$ and $Z_x^\prime=\tr [e^{-\beta_x H_{S_x}^\prime}]$ for $x=1,2$. Here the R\'enyi  $\alpha$-relative entropy is given by $D_{\alpha} (\rho \parallel \gamma )=\frac{\mbox{sgn}(\alpha)}{\alpha -1} \log \tr [\rho^{\alpha} \ \gamma^{1-\alpha}]$.
\end{thm}

The $\alpha$-free-entropies quantify the thermodynamic resource present in the system $S_{12}$ and it vanishes for the semi-Gibbs states. Therefore, any transformation among the block-diagonal states under the cSLTOs must respect the above monotonic relation for the $\alpha$-free-entropies for all $\alpha$. Note, the second laws leading to these monotonic relations can also be derived for more general state transformations as allowed by non-cyclic engine operations (see Supplementary Information). Apart from dictating state transformations, the Theorem \ref{thmmt:2ndLawsAll} delimits the amount of thermodynamic resource, i.e., free-entropy or work, can be extracted using a state transformation in an engine. It also quantifies the amount of the free-entropy required to be expended to make a transformation possible. For this, a bipartite battery $W_{12}$ is introduced that stores work in the form of pure energy. It is	attached with the system $S_{12}$ and then jointly evolved with cSLTOs, as shown in Fig.~\ref{fig:FreeEnt}. Now the free-entropy distance is introduced in Theorem \ref{thmmt:fed} to quantify the extractable free-entropy or the free-entropy cost in the one-shot finite-size regime, in terms of the works that can be stored in a battery. This in turn also quantifies the maximum extractable deterministic work from a quantum engine or the minimum deterministic work required to execute a refrigeration process.

\begin{figure}
\includegraphics[width=0.70\columnwidth]{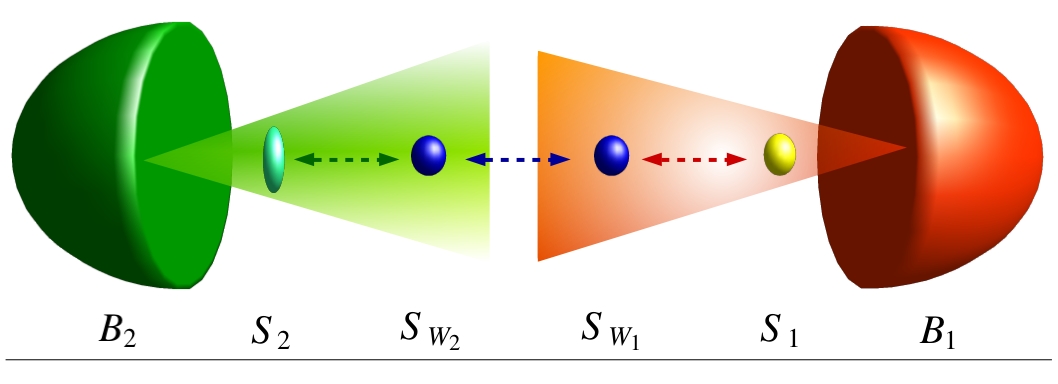}
\caption{{\bf Extraction of free-entropy.} A battery $S_{W_{12}}$, with two sub-systems $S_{W_1}$ and $S_{W_2}$ and the non-interacting Hamiltonian $H_{S_{W_{12}}}=H_{S_{W_1}} + H_{S_{W_2}}$ is attached with the system $S_{12}$ to store free-entropy (or work) once extracted. Without loss of generality, the battery subsystems are considered to be two-level systems with the Hamiltonians $H_{S_{W_{1}}}=W_{1} \proj{W_{1}}_{S_{W_{1}}}$ and $H_{S_{W_{2}}}=W_{2} \proj{W_{2}}_{S_{W_{2}}}$, and these are restricted to remain in the eigenstates of the Hamiltonians always. The $S_{W_1}$ is tagged with subsystem $S_1$ and similarly the $S_{W_2}$ is with $S_2$. The initial battery state is chosen to be the zero-energy state $ \rho_{S_{W_{12}}}^i=\proj{0}_{S_{W_1}} \otimes \proj{0}_{S_{W_2}}$. The composites $S_1S_{W_1}$ and $S_{2}S_{W_2}$ interact with the baths $B_1$ and $B_2$ (at different inverse temperatures $\beta_1$ and $\beta_2$) respectively through semi-local thermal operations, so that the overall transformation is $\left(\rho_{S_{12}} \otimes \rho_{S_{W_{12}}}^i, \ H_{S_{12}} + H_{S_{W_{12}}} \right) \rightarrow \left( \sigma_{S_{12}}^\prime \otimes \rho_{S_{W_{12}}}^f, \ H_{S_{12}}^\prime + H_{S_{W_{12}}}  \right)$, where final state of the battery is $\rho_{S_{W_{12}}}^f = \proj{W_1}_{S_{W_{1}}} \otimes \proj{W_2}_{S_{W_{2}}}$. Note, the battery Hamiltonian remains unchanged in the transformation. The values of $W_1$ and $W_2$ depend on the very cSLTO under which the transformation happens.}
\label{fig:FreeEnt}
\end{figure}

\begin{thm}[Free-entropy distance in each engine cycle]\label{thmmt:fed}
For the transformation in Eq. \eqref{eq:EngineTransformGen} via a cSLTO, the free-entropy distance between the initial and final states of the system is given by 
\begin{align}
 S_d(\rho_{12} & \rightarrow \sigma_{12}^\prime)= \beta_1 W_1 + \beta_2 W_2, \\
 &= \inf_{\alpha \geqslant 0}\left[S_\alpha(\rho_{S_{12}},\gamma_{S_1} \otimes \gamma_{S_2}) -  S_\alpha(\sigma_{S_{12}}^\prime,\gamma_{S_1}^\prime \otimes \gamma_{S_2}^\prime) \right]. \nonumber
 \end{align}
\end{thm} 
We refer to the Supplementary Information for the proof.
As per the second laws, if the initial state possesses larger free-entropy than the final one, i.e., $S_{\alpha}\left(\rho_{S_{12}}, \gamma_{S_{1}} \otimes \gamma_{S_{2}} \right) \geqslant S_{\alpha}\left(\sigma_{S_{12}}^\prime, \gamma_{S_{1}}^\prime \otimes \gamma_{S_{2}}^\prime \right)$ for all $\alpha \geqslant0$, the transformation can take place spontaneously under cSLTOs. This we term as the forward process. Then, the guaranteed one-shot extractable work from the process is 
\begin{align}\label{eq:extW}
W_{ext}=W_1+W_2 > 0, 
\end{align}
where $W_1 >0$ and $W_2$. To perform the reverse transformation $ \left( \sigma_{S_{12}}^\prime, H_{S_{12}}^\prime  \right) \rightarrow \left(\rho_{S_{12}}, H_{S_{12}} \right) $, as for a refrigeration process, the Theorem \ref{thmmt:fed} constrains that the minimum one-shot free-entropy to be supplied to ascertain the transformation is $S_d(\sigma_{12}^\prime \to \rho_{12})$. It is important to note that the free-entropy distance of a forward process is not in general equal to its reverse process, and $S_d(\rho_{12} \rightarrow \sigma_{12}^\prime) \leqslant - S_d(\sigma_{12}^\prime \to \rho_{12})$, where the equality holds for a few special cases. Therefore, the thermodynamic reversibility is no longer respected in the one-shot finite-size regime.

The one-step engine cycle also enables an engine to operate exclusively utilizing inter-system correlations. This is the second important result of this work. Consider that the subsystem $S_1$ (and $S_2$) is made up of two parties $M_1$ and $N_1$ (and $M_2$ and $N_2$) with Hamiltonian $H_{S_1}=H_{M_1} + H_{N_1}$ (and $H_{S_2}=H_{M_2} + H_{N_2}$). We further assume $H_{M_1}=H_{M_2}$ and $H_{N_1}=H_{N_2}$. The cycle starts with the state $\rho_{S_{12}}=\rho \otimes \sigma$, where the state $\rho$ of $S_1$ has no correlation between $M_1$ and $N_2$, i.e., $\rho= \rho_{M_1} \otimes \rho_{N_1}$ with $\rho_{M_1}=\tr_{N_1} \rho$ and $\rho_{N_1}=\tr_{M_1} \rho$. But, the state $\sigma$ of $S_2$ has non-vanishing correlation between $M_2$ and $N_2$, i.e., $\sigma \neq \sigma_{M_2} \otimes \sigma_{N_2}$ with $\sigma_{M_2}=\tr_{N_2} \sigma$ and $\sigma_{N_2}=\tr_{M_2} \sigma$. The one-step engine cycle that exclusively utilizes correlation is $\left(\rho_{S_{12}}, \ H_{S_{12}} \right) \to \left(\sigma_{S_{12}}, \ H_{S_{12}} \right)$, where $\sigma_{S_{12}}=\sigma \otimes \rho$. Then, in each cycle, the extractable free-entropy distance is $S_d(\rho_{S_{12}} \to \sigma_{S_{12}})=\beta_1 W_1^c + \beta_2 W_2^c$, and the extractable work is $W^c_{ext}=W_1^c + W_2^c$.

It is known that the inter-system correlation can store thermodynamic work potential and can lead to ``anomalous'' heat flow - a spontaneous heat transfer from a cooler to a warmer body \cite{Bera16}. However, the studies were restricted to the asymptotic regime. Now, we are able to characterize such thermodynamic potential and its role in anomalous heat flow in the one-shot finite-size regime. For instance, consider a system state $\rho_{S_{12}}$ which has non-vanishing correlation shared by the subsystems $S_1$ and $S_2$, i.e., $\rho_{S_{12}} \neq \rho_{S_1} \otimes \rho_{S_2}$, where $\rho_{S_1}=\tr_{S_2}[\rho_{S_{12}}]$ and $\rho_{S_2}=\tr_{S_1}[\rho_{S_{12}}]$. 
Then a (non-cyclic) state transformation $\left(\rho_{S_{12}}, \ H_{S_{12}} \right) \to \left(\rho_{S_1} \otimes \rho_{S_2}, \ H_{S_{12}} \right)$ leads to the free-entropy distance $S_d(\rho_{S_{12}} \to \rho_{S_1} \otimes \rho_{S_2})=\beta_1 W_1 + \beta_2 W_2$ and the extractable work $W_{ext}=W_1 + W_2$. In fact this work is responsible for the anomalous heat from or, equivalently, the refrigeration process. 

\subsection{Attaining Carnot efficiency}

The state transformations in the one-shot finite-size regime are very restrictive as these need to obey many second laws simultaneously. Moreover, thermodynamics is irreversible in general. This raises the question if there can be a heat engine that attains the maximum possible heat-to-work conversion efficiency (i.e., the Carnot efficiency) in the one-shot finite-size regime. Here we give an affirmative answer to this question that also vindicates the superiority of the quantum heat engines operating in the one-step cycle using semi-local thermal operations. The Carnot efficiency of a work extraction process is attained when the transformations occurring in an engine are reversible. In the one-shot finite-size regime, such a reversible engine transformation can be achieved if we consider initial and final states of the (working) system to be the energy eigenstates of the Hamiltonian. 

Let us present one such engine transformation that attains maximum possible heat-to-work conversion efficiency, i.e., the Carnot efficiency. Without loss of generality, we consider a two-qubit (working) system $S_{12}$ with the Hamiltonian $H_{S_{12}}=H_{S_1} + H_{S_2}$. The subsystem Hamiltonians are given by $H_{S_1}=a\ketbra{1}{1}_{S_1}$ and $H_{S_2}=a\ketbra{1}{1}_{S_2}$. The engine undergoes a one-step cycle following the transformation $(\rho_{S_{12}}, H_{S_{12}}) \to (\sigma_{S_{12}}, H_{S_{12}})$ using semi-local thermal operation, where the initial and final states respectively are 
\begin{align}
 \rho_{S_{12}}=\ketbra{0}{0}_{S_{1}} \otimes \ketbra{1}{1}_{S_{2}} \ \ \mbox{and} \ \ \sigma_{S_{12}}=\ketbra{1}{1}_{S_{1}} \otimes \ketbra{0}{0}_{S_{2}}.
\end{align}
The Hamiltonians of the subsystems $S_1$ and $S_2$ do not change under this transformation. The transformation is reversible because $\alpha$-free-entropies for such pure states are $\alpha$ independent. As a consequence, free-entropy distances satisfy $S_d(\rho_{S_{12}} \rightarrow \sigma_{S_{12}} ) = S_d(\rho_{S_{12}} \leftarrow \sigma_{S_{12}} )$, where 
\begin{align}\label{eq:SdCarnotMax}
 S_d(\rho_{S_{12}} \to \sigma_{S_{12}})=\beta_2 a - \beta_1 a = \beta_1 W_1 + \beta_2 W_2 > 0.
\end{align}
The net extracted work in each engine cycle is given by $W_{ext}=W_1 + W_2 >0$. Recall, $\beta_1 < \beta_2$.

To understand the conversion of heat into work and how it attains the Carnot efficiency, let us explore the state transformation in the engine considering baths-system-battery composite all together. Since the system-battery $S_{12}S_{W_{12}}$ is in a block-diagonal state, the initial state of baths-system-battery composite can be expressed in the block-diagonal form as 
\begin{align}
 \gamma_{B_1} \otimes \gamma_{B_2} \otimes \rho_{S_{12}} \otimes \rho_{S_{W_{12}}}=\bigoplus_{E_1 + E_2} [\gamma_{B_1} \otimes \gamma_{B_2} \otimes \rho_{S_{12}} \otimes \rho_{S_{W_{12}}}]_{E_1 + E_2}, \nonumber
\end{align}
where $[\gamma_{B_1} \otimes \gamma_{B_2} \otimes \rho_{S_{12}} \otimes \rho_{S_{W_{12}}}]_{E_1 + E_2}$ is the part of the global state that resides in the total energy block $E_1 + E_2$. Here $E_1=E_{S_1}+E_{B_1}$ is the sum of energies belonging to $S_1$ and $B_1$, and similarly for $E_2=E_{S_2}+E_{B_2}$. A global unitary $U$ is applied jointly on the composite conserves total energy and the total weighted-energy of the baths-system composite due to the constraints \eqref{eqmt:slto-commutation} and \eqref{eqmt:slto-commutationTemp} and has the block-diagonal structure, given by $U=\bigoplus_{E_1 + E_2} U_{E_1 + E_2}$ (see Supplementary Information). The unitary $U_{E_1 + E_2}$ applies to the total energy block $E_1 + E_2$ and is allowed transfer populations within the block so that total weighted-energy $\beta_1 E_1 + \beta_2 E_2$ is strictly conserved. In the block of total energy $E_1+E_2$, the transformation becomes
\begin{align}
[\gamma_{B_1} \otimes \gamma_{B_2} \otimes \rho_{S_{12}}]_{E_1+E_2} \otimes  \rho^i_{S_{W_{12}}} \to [\sigma_{B_1B_2S_{12}}]_{E_1^\prime + E_2^\prime} \otimes \rho^f_{S_{W_{12}}}, \nonumber
\end{align}
where $\rho_{S_{W_{12}}}^i=\ketbra{0}{0}_{S_{W_1}} \otimes \ketbra{0}{0}_{S_{W_2}}$ is the initial battery state, $\rho_{S_{W_{12}}}^f=\ketbra{W_1}{W_1}_{S_{W_1}} \otimes \ketbra{W_2}{W_2}_{S_{W_2}}$ is the final battery state, and $\tr_{B_1B_2} \sigma_{B_1B_2S_{12}}=\sigma_{S_{12}}$. Here the bipartite battery Hamiltonian is $H_{S_{W_{12}}}=H_{S_{W_{1}}}=W_1 \ketbra{1}{1}_{S_{W_1}} + W_2 \ketbra{1}{1}_{S_{W_2}}$. Now the conservation of the total weighted-energy and the total energy ensure that 
\begin{align}
 \beta_1 E_1 + \beta_2 E_2 & = \beta_1 (E_1^\prime + W_1) + \beta_2 (E_2^
\prime + W_2), \label{eq:TotWeightEngCons} \\ 
 E_1 + E_2 & =E_1^\prime +  E_2^\prime + W_1 + W_2, \label{eq:TotEngCons}
\end{align}
where $E_1=E_{B_1}$, $E_2=E_{B_2}+a$, $E_1^\prime=E_{B_1}^\prime + a$, and $E_2^\prime=E_{B_2}^\prime$. The Eqs.~\eqref{eq:TotWeightEngCons} and \eqref{eq:TotEngCons}, along with Eq.~\eqref{eq:SdCarnotMax}, lead to
\begin{align}
\beta_1 (E_{B_1} - E_{B_1}^\prime) + \beta_2 (E_{B_2} - E_{B_2}^\prime) & = \beta_1 Q_1 + \beta_2 Q_2=0, \label{eq:ClausiusSaturate} \\ 
W_{ext}=W_1+W_2 & =Q_1 + Q_2 >0, \label{eq:maxWork} 
\end{align}
where we have identified the heat as the change in energy of the bath $B_1$ given by $Q_1=E_{B_1} - E_{B_1}^\prime$ and similarly $Q_2=E_{B_2} - E_{B_2}^\prime$ for bath $B_2$. The Eq.~\eqref{eq:ClausiusSaturate} represents the Clausius equality for the cyclic process. The Eq.~\eqref{eq:maxWork} quantifies the extracted work in each one-step cycle. The other total energy blocks will result in identical Clausius equality and the same amount of extracted work. Therefore, the efficiency for the heat-to-work conversion is now reduced to
\begin{align}
 \eta_C=\frac{W_{ext}}{Q_1}=1-\frac{\beta_1}{\beta_2},
\end{align}
which is exactly the Carnot efficiency. As a result, the Carnot's efficiency can be attained for an engine operating in the one-shot finite-size regime. A refrigerator can also be constructed just by reversing engine cycle, i.e., $(\sigma_{S_{12}}, H_{S_{12}}) \to (\rho_{S_{12}}, H_{S_{12}})$, and that achieves maximum possible coefficient of performance in the one-shot finite-size regime. The example for the two-qubit system we have considered above can automatically be generalized to arbitrary dimensional (working) systems where the subsystems are also allowed to update their Hamiltonians. The only requirement to attain the Carnot efficiency is that the initial and final states of the system have to be in the eigenstates of the initial and final Hamiltonians.

\section{Discussion}

In this work, we have designed a heat engine that attains Carnot efficiency in the one-shot finite-size regime. These engines are superior in work extraction compared with the other engines so far considered in the literature. To prove our result, we have developed a resource theory of quantum heat engines to systematically study thermodynamics in the presence of two or more baths in the finite-size one-shot regime (see Supplementary Information). We stress that the earlier works focus on how the a-thermal (or out-of-equilibrium) properties of a system transform into thermodynamic work and do not consider the conversion of heat into work. For such considerations, one thermal bath is sufficient. On the contrary, the resource theory developed for quantum heat engines, with two or more thermal baths, is the only framework to systematically address how and to what extent the heat can be converted into work in the quantum heat engines operating in the one-shot finite-size regime.

We have proceeded with the precise characterizations of thermodynamic operations, i.e., the semi-local thermal operations, by introducing the first law for engines, where the system simultaneously interacts with both the baths. In addition to strict total energy conservation, the first law also ensures a strict weighted-energy conservation. Then, we have defined semi-Gibbs states as resource-free states and free-entropies as the measure of thermodynamic resource applicable in the one-shot finite-size regime. With this, we have formulated the second laws for state transformations in a quantum engine. Our formulation reveals that the state transformations in a quantum engine are irreversible, in general, and are dictated by many second laws (see Supplementary Information). In addition, the framework enables us to design a Carnot heat engine operating with a one-step cycle that can enhance work extraction efficiency in the one-shot finite-size regime (see Method). Most importantly, we have presented the reversible engine operations that result in the maximum possible heat-to-work conversion efficiency, i.e., the Carnot efficiency, in the one-shot finite-size regime. 

Although the framework for resource theory of heat engines is mathematically rigorous and clean in the theoretical sense, the semi-local thermal operations are difficult to implement in experiments. Note, there have been various proposals exploring possible physical realizations of (local) thermal operations in the presence of a single bath, see for example \cite{Mazurek18, Perry18, Lostaglio2018}. Following a similar track, it can also be possible to implement semi-local thermal operations. In particular, the proposed engine operation that results in higher efficiency in the conversion of heat into work can be implemented without much difficulties. This will certainly open up the possibility to experimentally realize quantum heat engines operating in a one-step cycle with higher one-shot efficiencies in the finite-size regime, even a possibility of attaining Carnot efficiency.

In summary, our work achieves:
\begin{itemize}
\item An engine with a one-step cycle leading to higher heat-to-work conversion efficiency, even attaining Carnot efficiency in the one-shot finite-size regime for the first time. The results in turn demonstrate the superiority of the proposed engines over other engines operating in the one-shot finite-size regime.

\item A quantum heat engine that converts heat into work by exclusively utilizing inter-system correlations.

\item A fundamental gain, i.e., a concrete mathematical framework leading to a resource theory and a novel theoretical understanding of quantum and nano-scale heat engine, and, in particular, the conversion of heat into work in quantum heat engines operating in the one-shot finite-size regime, and the role of inter-system correlations in such processes.
 
\item Possibilities of new experimental avenues for quantum heat engines that attain higher or even Carnot efficiency in the one-shot finite-size regime. 
\end{itemize} 

\

\section{Method}
\subsection{Characterization of the semi-local thermal operations}
The SLTOs, introduced in the Definition~\ref{propmt:slto}, possess several important properties. We outline them below.  \\

\noindent 
(P1) For an arbitrary initial semi-Gibbs state $\gamma_{S_{12}} \in \mathcal{T}_{S_{12}}$, the SLTOs satisfy  $\Lambda_{S_{12}} \left(\gamma_{S_{12}} \right) \in \mathcal{T}_{S_{12}} $. Therefore, the SLTOs map semi-Gibbs states onto itself. This is a consequence of the very definition of operation in Definition \ref{propmt:slto}, in particular the constraint \eqref{eqmt:slto-commutationTemp}. If the local states of the subsystems, $S_1$ and $S_2$ are in thermal equilibrium, the SLTOs cannot bring them away from their local equilibrium, despite the fact that the SLTOs are semi-local in nature and allow an exchange of energy among the subsystems. This happens despite the fact that the operation is global in nature and justifies the semi-local character of the SLTOs. Therefore, under the SLTOs, the semi-Gibbs states represent the fixed points. This is one of the required properties the allowed operations should possess to formulate a resource theory.  \\

\noindent
(P2) Action of an operation that satisfies $\Lambda_{S_{12}} \left(\gamma_{S_{12}} \right) \in \mathcal{T}_{S_{12}}$, on a system state $\rho_{S_{12}}$ can be simulated by an SLTO, given that $[\rho_{S_{12}}, H_{S_{12}}]=0$. A detailed proof is given in the Supplementary Information. Note, these operations are termed as the semi-Gibbs preserving operations and form a larger set of operation than that of the SLTOs.   \\

\noindent 
(P3) The SLTOs can implement the changes in the system Hamiltonians. For instance, an initial non-interacting Hamiltonian of $S_{12}$ can be updated to a new one, as $H_{S_{12}}=H_{S_1} + H_{S_2} \to H_{S_{12}}^\prime=H_{S_1}^\prime + H_{S_2}^\prime$, with the help of clocks (see Supplementary Information). These operations are nothing but the (semi-local) adiabatic transformations in a typical heat engine. \\ 

\noindent
(P4) It is interesting to note that the SLTOs, as well as the cSLTOs, are time-translation symmetric operations with respect to the time evolution generated by the Hamiltonian $H_{S_{12}}=H_{S_1} + H_{S_2}$. In other words, if there is a time translation of the system driven by unitary $V(t)=e^{-itH_{S_{12}}/\hbar}$ for any given time $t$, then 
\begin{align}\label{eq:ttsymtotalHam12}
 V(t) \left[\Lambda_{S_{12}}(\rho_{S_{12}})\right]V^\dag(t)= \Lambda_{S_{12}}\left[V(t) (\rho_{S_{12}})V^\dag(t)\right],
\end{align}
i.e., the order of the time translation operations and the SLTO commute. Because of this symmetric property, the SLTOs monotonically decrease the superpositions between different eigenbases of $H_{S_{12}}$ or, in other words, among the energy eigenbases. Note the SLTOs are also time-translation symmetric with respect to the system's weighted-Hamiltonian $H^{\beta_1 \beta_2}_{S_{12}}=\beta_1 H_{S_1} + \beta_2 H_{S_2}$, since $[H_{S_{12}}^{\beta_1 \beta_2}, H_{S_{12}}]=0$. 

The Eq.~\eqref{eq:ttsymtotalHam12} further implies that the SLTOs commute with the dephasing operations in the eigenbases of the Hamiltonian of the system $H_{S_{12}}=\sum_{i,j} (E^{S_1}_i + E^{S_2}_j) \proj{ij}$, i.e., 
\begin{align}
\Lambda_{S_{12}}\circ P_{S_{12}} (\rho_{S_{12}}) = P_{S_{12}} \circ \Lambda_{S_{12}} (\rho_{S_{12}}), \label{eq:dephaseGlob}
\end{align}
where $P_{S_{12}}(\rho_{S_{12}})= \sum_{ij} \braket{ij|\rho_{S_{12}}|ij} \ketbra{ij}{ij}$ is the dephasing operation. Note the dephasing operation can be achieved by time averaging time-translated state for a long enough time $T$,
\begin{align}
P_{S_{12}}(\rho_{S_{12}})=\frac{1}{T} \int_0^T V(t) \ (\rho_{S_{12}}) \ V(t) dt. \label{eqApp:DephasedStateWeightedHamil}
\end{align}
The Eq. \eqref{eq:dephaseGlob} signifies that the diagonal elements in the eigenbases of $H_{S_{12}}$, i.e., $P_{S_{12}}(\rho_{S_{12}})$, evolve independently of the off-diagonal elements. Further, the amount of asymmetry present in a state due to the superposition between different energy eigenbases monotonically decreases under the SLTOs. We use these properties to understand the free-entropy locking in superpositions and to add more conditions to supplement the second laws for state transformations (see Supplementary Information).

\subsection{More extractable work using cSLTOs}
Here we show how the semi-local character in the cSLTOs can lead to a higher amount of extractable work compared to the engine that operates using (local) thermal operations. Consider that the working system is initially in an uncorrelated state $\rho_{S_{12}}=\rho \otimes \sigma$ in an engine. Then the one-step engine transformation, given in Eq.~\eqref{eq:EngineTransformGen}, reduces to 
\begin{align}
\label{eq:EngCycleUnCorrmt}
 (\rho \otimes \sigma, \ H_{S_{12}})  \rightarrow (\sigma \otimes \rho, \ H_{S_{12}}^\prime).
\end{align} 
In this transformation, there are two sub-transformations happening simultaneously via a cSLTO; (i) forward sub-transformation, $(\rho, H_{S_1}) \to (\sigma, H_{S_2})$ in presence of the bath $B_1$ at inverse temperature $\beta_1$, and (ii) the reverse sub-transformation $(\sigma, H_{S_2}) \to (\rho, H_{S_1})$ while interacting with bath $B_2$ at inverse temperature $\beta_2$. 

For the uncorrelated state $\rho_{12}=\rho \otimes \sigma$, the $\alpha$-free-entropy becomes additive $S_\alpha (\rho \otimes \sigma, \gamma_{S_1}\otimes \gamma_{S_2})=S_\alpha (\rho, \gamma_{S_1})+S_\alpha (\sigma,  \gamma_{S_2})$, where $S_\alpha (\rho_{S_{x}}, \gamma_x)= D_\alpha \left(\rho_{S_{x}} \parallel \gamma_{S_x} \right) - \log Z_x$. The $S_\alpha (\rho_{S_{x}}, \gamma_x)/\beta_x $ is the accessible $\alpha$-free energy stored in the system $\rho_{S_{x}}$ and can be converted into work using a (local) thermal operation in presence of a bath at inverse temperature $\beta_x$ \cite{Brandao15}. The second laws, considering the sub-transformations (i) and (ii) simultaneously occur via a cSLTO, imply 
\begin{align}\label{eq:2ndLawsWEng}
\beta_1 W_1^{(\alpha)} + \beta_2  W_2^{(\alpha)} \geqslant 0, \ \ \forall \alpha \geqslant 0,
\end{align} 
where the $W_x^{(\alpha)}$ quantifies the change in $\alpha$-work due to the transformation in the presence of the bath at inverse temperature $\beta_x$. In terms of the $\alpha$-free energies \cite{Brandao15}, we express these $\alpha$-works as
\begin{align}
 & W_1^{(\alpha)} = \frac{1}{\beta_1} \left[ S_{\alpha} \left(\rho, \gamma_{S_{1}}  \right) - S_{\alpha} \left(\sigma, \gamma_{S_{1}}^\prime \right) \right], \\
 & W_2^{(\alpha)}= \frac{1}{\beta_2} \left[ S_{\alpha} \left(\sigma, \gamma_{S_{2}}  \right) - S_{\alpha} \left(\rho, \gamma_{S_{2}} ^\prime \right) \right],
\end{align}
where $\gamma_{S_{x}}=\frac{e^{-\beta_x H_{S_x}}}{\tr [e^{-\beta_x H_{S_x}}]}$ for $x=1,2$, $\gamma_{S_1}^\prime=\frac{e^{-\beta_1 H_{S_2}}}{\tr [e^{-\beta_1 H_{S_2}}]}$, and $\gamma_{S_2}^\prime=\frac{e^{-\beta_2 H_{S_1}}}{\tr [e^{-\beta_2 H_{S_1}}]}$. Given $\beta_1 < \beta_2$ and a spontaneous engine cycle, the Eq.~\eqref{eq:2ndLawsWEng} guarantees that $W_{ext}^{(\alpha)}=W_1^{(\alpha)} + W_2^{(\alpha)} > 0, \ \forall \alpha \geqslant 0$. 

Let us now show that the one-shot efficiency of the Carnot engine operating via cSLTOs is larger, in general, compared to the case considered in Figure \ref{fig:QHEng}, where the system locally interacts with individual baths at a time. Suppose that the system locally interacts with the baths using local thermal operations \cite{Brandao13, Horodecki13, Brandao15} and undergoes two sub-transformations (i) and (ii) in separate steps to complete the Carnot cycle, as discussed earlier. For these sub-transformations (i) and (ii), the one-shot extractable work and the work cost under local thermal operations, respectively, are
\begin{align}
\bar{W}_1=\inf_{\alpha \geqslant 0}[W_1^{(\alpha)}] \leqslant W_1, \ \ \ \mbox{and} \ \ \  
\bar{W}_2=\sup_{\alpha \geqslant 0}[W_2^{(\alpha)}] \geqslant W_2.
\end{align}
The net one-shot extracted work using local thermal operations is $\bar{W}_{ext}=\bar{W}_1 + \bar{W}_2$, where $\bar{W}_1 \geqslant 0$ and $\bar{W}_2 \leqslant 0$. It is easy to check that $W_{ext} \geqslant \bar{W}_{ext}$ which is satisfied for arbitrary engine cycle in general. \\


\noindent {\bf Acknowledgments:} The authors thank Andreas Winter for fruitful discussions, and also thank Markus P. M\"uller, Mischa P. Woods, and Raam Uzdin for useful comments. ML and MLB acknowledge supports from ERC AdG NOQIA, Spanish Ministry of Economy and Competitiveness (``Severo Ochoa'' program for Centres of Excellence in R \& D (CEX2019-000910-S), Plan National FIDEUA PID2019-106901GB-I00/10.13039 / 501100011033, FPI), Fundació Privada Cellex, Fundació Mir-Puig, and from Generalitat de Catalunya (AGAUR Grant No. 2017 SGR 1341, CERCA program, QuantumCAT U16-011424 , co-funded by ERDF Operational Program of Catalonia 2014-2020), MINECO-EU QUANTERA MAQS (funded by State Research Agency (AEI) PCI2019-111828-2 / 10.13039/501100011033), EU Horizon 2020 FET-OPEN OPTOLogic (Grant No 899794), the National Science Centre, Poland-Symfonia Grant No. 2016/20/W/ST4/00314, and the Spanish Ministry MINECO. MNB gratefully acknowledges financial supports from SERB-DST (CRG/2019/002199), Government of India. 




%

\clearpage

\onecolumngrid
 \begin{center}
   \textbf{\large \\Supplementary Information: \\ 
   	Attaining Carnot Efficiency with Quantum and Nano-scale Heat Engines}\\[.2cm]
   Mohit Lal Bera,$^{1}$ Maciej Lewenstein,$^{1,2}$ and Manabendra Nath Bera$^{3,*}$\\[.1cm]
   {\itshape ${}^1$ICFO -- Institut de Ci\`encies Fot\`oniques, The Barcelona Institute of Science and Technology, ES-08860 Castelldefels, Spain\\
   ${}^2$ICREA, Pg.~Lluis Companys 23, ES-08010 Barcelona, Spain\\
   ${}^3$Department of Physical Sciences, Indian Institute of Science Education and Research (IISER), Mohali, Punjab 140306, India\\}
   ${}^*$Electronic address: mnbera@gmail.com\\
 \end{center}

\setcounter{equation}{0}
\setcounter{figure}{0}
\setcounter{table}{0}
\setcounter{page}{1}
\setcounter{section}{0}
\renewcommand{\theequation}{S\arabic{equation}}
\renewcommand{\figurename}{Supplementary Figure}
\renewcommand{\tablename}{Supplementary Table}


\section{Background and motivations for resource theory of quantum heat engines \label{sec:Back}}
In the last decades, enormous efforts have been put forward to extend thermodynamics to the regimes where a system made up of a finite (typically moderate or small) number of quantum particles interacts with a single thermal bath at fixed temperature \cite{Gemmer09, Binder18}. This also includes the situation where one has access to repeated, simultaneous, and one-shot measurements on the particles. It leads to two major approaches to studying quantum thermodynamics. The first, which applies in fact to both the asymptotic regime and the many-shot finite-size regime, is based on fluctuation theorems (FT), exploiting statistical mechanics and open quantum systems dynamics \cite{Jarzynski97, Crooks99, Campisi11}. The other one is based on the quantum information theory \cite{Brandao13, Skrzypczyk14, Horodecki13, Aberg13, Brandao15, Lostaglio15, Cwiklinski15, Lostaglio15a, Bera16, Bera17, Sparaciari17, Uzdin18, Muller18}. Among others, the latter leads to a resource theory of quantum systems out of thermal equilibrium, which is commonly termed as the resource theory of quantum thermodynamics (RTQTh) \cite{Brandao13, Horodecki13, Brandao15}. The RTQTh is applicable to the asymptotic regime, and both many-shot and one-shot finite-size regimes. The RTQTh stands out among the other approaches as it exploits a rigorous mathematical framework similar to the resource theory of entanglement, where the latter was developed to characterize the role of entanglement in quantum information processing. There is also another formulation of a resource theory based on complete-passivity (CPTh) \cite{Sparaciari17, Bera17} that generalizes thermodynamics to the situation where system and baths become comparable in size. But, the CPTh is applicable to the asymptotic regime only. These different approaches can be classified in terms of their applications in different regimes, as given in Supplementary Table \ref{tab:classification}.

\begin{table}[h]
	\begin{center}
		\begin{tabular}{ |c | c | c| }
			\hline
			{\bf Regimes} & {\bf Asymptotic} ($N \to \infty$)  & {\bf Finite-size} ($N \ll \infty$) \\
			\hline
			{\bf Repeated} &  &  \\
			{\bf measurements} & STh, \ CPTh, \ FT, \ RTQTh & FT, \ RTQTh \\
			\hline
			{\bf One-shot} &  &  \\
			{\bf measurement } & FT, \ RTQTh & RTQTh \\
			\hline  
		\end{tabular}
		\caption{\label{tab:classification} The Supplementary table classifies various approaches based on their regime of applications. Here STh, FT, CPTh, and RTQTh represent the standard thermodynamics, the fluctuation theorem, the complete-passivity based resource theory of thermodynamics, and the resource theory of quantum thermodynamics respectively. The $N$ denotes the number of particles in a system interacting with a thermal bath. Note, RTQTh fully characterizes quantum thermodynamics only for the systems that are block-diagonal in the energy eigenbases \cite{Brandao15}.}
	\end{center}
\end{table}

The resource theoretic formulation reveals that thermodynamics in the one-shot finite-size regime is not reversible and one needs many second laws, associated with many free energies, to characterize the transformations among the states that are block-diagonal in energy eigenbases \cite{Horodecki13, Aberg13, Brandao15}. These second laws have been further studied for more general states having superpositions in energy eigenbases \cite{Lostaglio15, Cwiklinski15, Lostaglio15a, Gour2018}. Interestingly, in \cite{Muller18}, it has been shown that by allowing a non-vanishing amount of correlation all these many second laws can be reduced to a single one, based on standard Helmholtz free energy. Recently, the approaches based on fluctuation theory and resource theory have been inter-connected for some cases \cite{Alhambra16, Aberg18, Guarnieri19}. However, all these investigations are limited to the situations, where the quantum system is interacting with only \emph{one} thermal bath at a fixed temperature. Therefore it is a natural question to ask whether it is possible to formulate a \emph{resource theory} for heat engines operating in the \emph{one-shot finite-size regime}, where a system composed of few quantum particles is interacting with two or more thermal baths at different temperatures. Apart from some efforts to quantify extractable work and engine efficiency in few special cases, and to study the finite-size effects and the quantum signatures \cite{Alicki79, Kieu04, Verley14, Rosnagel14, Uzdin15, Tajima17, Ng17, Ito18, Woods2019, Woods19a}, there has been no major progress, so far, in formulating a resource theory for quantum heat engines. 

The goal of this work is to formulate such a resource theory. It is worth mentioning that much of the earlier works, applied to the one-shot finite-size regime, focus on how the a-thermal (non-equilibrium) property of a system can be converted into thermodynamic work, and, for that, one thermal bath is enough. On the contrary, here we develop a resource theory of quantum heat engines to address how, and to what extent, the heat can be converted into work in the one-shot finite-size regime. Therefore, the theory is fundamentally different from the one considered in the earlier works. The new formalism provides the foundation for a novel theoretical understanding of the one-shot conversion of heat into work and the role of inter-system correlations in such processes. At the same time, it opens up new avenues to explore physically realizable quantum heat engines that have higher efficiency and attain Carnot efficiency in the one-shot finite-size regime. 

For any resource theory, it is necessary to identify the resource-free states that possess zero thermodynamic resources, and the free operations that neither can create (or increase) resource in a state nor requires resource to implement. Further, for a resource theory of quantum heat engines, the resource should, in principle, be related to the thermodynamic work. In the resource theory of thermodynamics in the presence single bath at a fixed temperature \cite{Brandao13, Horodecki13, Brandao15}, the resource-free states are the ones that are in thermal equilibrium with the bath, often termed as the Gibbs states. The free operations are the (local) thermal operations, and the thermodynamic resource is quantified by the free energies. The framework mostly focuses on how the a-thermal (out-of-equilibrium) properties of a system transform into thermodynamic work.

The situation changes drastically once one considers thermodynamics in the presence of more than one thermal bath at different temperatures, which is the case for a heat engine. The first difficulty appears in defining the resource-free states. There does not exist a state that is simultaneously in equilibrium with all the baths. All states have some non-vanishing thermodynamic resources. Furthermore, it is not possible to define the free operations, as the free operations are supposed to map a resource-free state to a resource-free state. Therefore, one cannot formulate a resource theory for heat engines just by merely extending the one formulated for a single bath. Rather, to start with, it requires one to introduce a new class of thermodynamic operations that are allowed in a heat engine, a new form of states as the resource-free states, and to invoke new quantifiers of thermodynamic resources. This leads us to introduce an entirely new resource theory for quantum and nano-scale heat engines, below.

\section{Thermal baths and system-bath composites}
The goal of this section is to characterize the Hilbert spaces of considerably large bath(s) at certain temperature(s), small systems that are in and away from thermal equilibrium, and their composites.

\subsection{Some useful properties of baths}
There are several useful properties of a considerably large baths, compared to the systems they interact with. A bath is considered to be always in thermal equilibrium at a fixed temperature, even after it interacts with a system. Therefore it has to be reasonably large so that it almost does not change after the interaction and remain in equilibrium. So, a bath being large is an important assumption.

All the systems, we consider, have Hamiltonians bounded from below, i.e., the lowest energy is zero. Consider a bath $B_x$ the Hamiltonian $H_{B_x}$ which has the largest energy is $E^{max}_{B_x} \rightarrow \infty$. The heat bath always remains in a Gibbs state $\gamma_{B_x}=\frac{e^{-\beta H_{B_x}}}{\tr [e^{-\beta H_{B_x}}]}$ with inverse temperature $\beta_x$. Now say there two baths $B_1$ and $B_2$ with the Hamiltonians $H_{B_1}$ and $H_{B_2}$ and the inverse temperatures $\beta_1$ and $\beta_2$. The joint thermal state of the baths is expressed as
\begin{align}
	\gamma_{B_{12}}=\gamma_{B_1} \otimes \gamma_{B_2}.
\end{align}
There exists a set of energies $\mathcal{E}_{B_{12}}$ in which the baths jointly live with high probability. Mathematically, for the projector $P_{\mathcal{E}_{B_{12}}}$ that spans over the space with a set of the total energies $\mathcal{E}_{B_{12}}$, this is expressed as 
\begin{align}\label{eq:BathTypicality}
	\tr [P_{\mathcal{E}_{B_{12}}} \gamma_{B_{12}}] \geqslant 1- \delta,
\end{align}
where $\delta > 0$. Given this, the bath satisfies the following properties, (cf. \cite{Horodecki13}):
\begin{itemize}
	\item The energy $E_{B_{12}} \in \mathcal{E}_{B_{12}}$ is peaked around a mean value as $E_{B_{12}} \in \left\{ \braket{E_{B_{12}}} - O(\sqrt{E_{B_{12}}}), \ldots , \braket{E_{B_{12}}} + O(\sqrt{E_{B_{12}}})  \right\}$.  
	
	\item The degeneracies $g_B(E_{B_{12}})$ in the energies $E_{B_{12}}= E_{B_{1}} + E_{B_{2}}\in \mathcal{E}_{B_{12}}$ scale exponentially with $E_{B_{1}}$ and $E_{B_{1}}$, i.e., $g_B(E_{B_{12}}) \geqslant e^{x E_{B_{1}} + y E_{B_{2}}}$, where $x, \ y$ are constants. Here $E_{B_{1}}$ and $E_{B_{2}}$ are the energies of the baths $B_1$ and $B_2$ respectively.
	
	\item Consider any pair of three energies $(E_{B_{1}}, E_{S_{1}}$, $E_{S_{1}}^\prime)$ and $(E_{B_{2}}, E_{S_{2}}$, $E_{S_{2}}^\prime)$, so that $E_{B_{12}}= E_{B_{1}} + E_{B_{2}} \in \mathcal{E}_{B_{12}}$, $E_{S_1} \ll E_{B_1}$, and $E_{S_1}^\prime \ll E_{B_1}$, and similarly $E_{S_2} \ll E_{B_2}$, and $E_{S_2}^\prime \ll E_{B_2}$. Then there exists a $E_{B_{12}}^\prime= E_{B_{1}}^\prime + E_{B_{2}}^\prime \in \mathcal{E}_{B_{12}}$ so that $E_{B_{1}} + E_{S_{1}}=E_{B_{1}}^\prime + E_{S_{1}}^\prime$ and $E_{B_{2}} + E_{S_{2}}=E_{B_{2}}^\prime + E_{S_{2}}^\prime$.

	\item For an energy $ E_{B_{12}} \in \mathcal{E}_{B_{12}}$, the degeneracies satisfy $g_B(E_{B_{12}} + E_{S_{1}} + E_{S_{2}}) \approx g_B(E_{B_{12}}) e^{\beta_1 E_{S_{1}}+\beta_2 E_{S_{2}}}$.
\end{itemize}
These properties are instrumental in understanding the thermodynamics of quantum and nano-scale systems interacting with large baths.

\subsection{Two baths and two (sub-)systems \label{sec:2bath2sys}}
Without loss of generality we consider a bipartite system $S_{12}$ with two subsystems, $S_1$ and $S_2$, that are \emph{semi-locally} interacting with two baths $B_1$ and $B_2$ respectively where the baths are with the inverse temperatures $\beta_1$ and $\beta_2$. We skip the discussion on the notion of ``semi-local'' here. We elaborate on it later to characterize the thermodynamics operations that are applicable in a quantum heat engine.

\begin{figure}[h]
	\includegraphics[width=0.6\columnwidth]{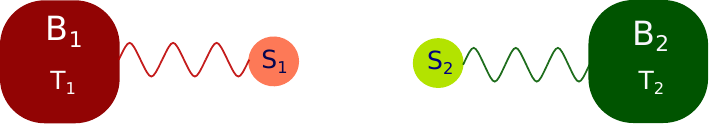}
	\caption{A schematic of a situation where system $S_1$ is (semi-locally) interacting with the bath $B_1$, and the system $S_2$ is semi-locally interacting with the bath $B_2$. Temperatures of the baths are $T_1=\frac{1}{\beta_1}$ and $T_2=\frac{1}{\beta_2}$.
		\label{fig:QHEcopy1}}
\end{figure} 

Say, the two considerably large baths $B_1$ and $B_2$ are with Hamiltonians $H_{B_1}$ and $H_{B_2}$ respectively. Further, the small systems $S_1$ and $S_2$ are the subsystems of a bipartite system $S_{12}$ with the Hilbert space $\mathcal{H}_{S_{12}}=\mathcal{H}_{S_1} \otimes \mathcal{H}_{S_2}$. The system $S_{12}$ possesses a non-interacting Hamiltonian $H_{S_{12}}=H_{S_1} + H_{S_2}$. We denote $E_{B_x}$ and $E_{S_x}$ as the energies of the bath $B_x$ and the subsystem $S_x$ respectively. The bath Hamiltonians $H_{B_{1/2}}$ are bounded from below and could have $E_{B_{1/2}}^{\max} \rightarrow \infty$. The system Hamiltonians $H_{S_{1/2}}$ are also bounded from below and satisfy $E_{S_{1/2}}^{\max} \ll E_{B_{1/2}}^{\max}$. We consider here non-degenerate system Hamiltonians $H_{S_{1/2}}$. Extension to degenerate cases can be done easily.

The underlying joint Hilbert space corresponding to the $S_{12}$, $B_1$, and $B_2$ is now $\mathcal{H}_{B_1} \otimes \mathcal{H}_{B_2} \otimes \mathcal{H}_{S_1} \otimes \mathcal{H}_{S_2}$. Here we assume the systems to interact as in Supplementary Figure \ref{fig:QHEcopy1}. It can be expressed as the Kronecker sums of constant total energy sub-spaces, i.e.,
\begin{align}
	\mathcal{H}_{B_1} \otimes \mathcal{H}_{B_2} \otimes \mathcal{H}_{S_1} \otimes \mathcal{H}_{S_2} =\bigoplus_{E_1 + E_2} \left(\bigoplus_{E_{S_1}+ E_{S_2}} \mathcal{H}^{E_1-E_{S_1}}_{B_1} \otimes \mathcal{H}^{E_2-E_{S_2}}_{B_2} \otimes \mathcal{H}_{S_1}^{E_{S_1}} \otimes \mathcal{H}_{S_2}^{E_{S_2}} \right), 
\end{align}
where the total energies are given by $E_1 + E_2$, and $E_1 = E_{B_1}+E_{S_1}$ and $E_2 = E_{B_2}+E_{S_2}$ are the energies corresponding to the composites $B_1S_1$ and $B_2S_2$ respectively. The total energies are the eigenvalues of the total Hamiltonian 
\begin{align}
	H_{B_1B_2S_{12}}=H_{B_1} + H_{S_1} + H_{B_2} + H_{S_2}.
\end{align}
It is important to notice that there are many combinations $E_1$ and $E_2$ for which $E_1 + E_2$ is identical. Consequently, any system-baths joint state can be written in terms of fixed total energy blocks. A system-baths state $\gamma_{B_1} \otimes \gamma_{B_2} \otimes \rho_{S_{12}}$, which is diagonal in the energy eigenbases, can be expressed as
\begin{align}\label{eq:WeightedEngBlocksState}
	\gamma_{B_1} \otimes \gamma_{B_2} \otimes \rho_{S_{12}}=\sum_{E_1 + E_2} P_{E_1 + E_2} \left( \gamma_{B_1} \otimes \gamma_{B_2} \otimes \rho_{S_{12}} \right) P_{E_1 + E_2}=\sum_{E_1 + E_2} p_{E_1 + E_2} \  \rho_{E_1 + E_2}^{B_1B_2S_{12}},
\end{align}
where $P_{E_1 + E_2}$s are the projectors with total energy $E_1 + E_2$ and $p_{E_1 + E_2}=\tr \left[ P_{E_1 + E_2} \left( \gamma_{B_1} \otimes \gamma_{B_2} \otimes \rho_{S_{12}} \right) \right]$ are the probabilities. Note, for a given value $E_1 + E_2=X$, the projector is expressed as $P_{E_1 + E_2}=\sum_{E_1^{i}, E_2^{j}} P_{E_1^{i}} \otimes P_{E_2^{j}}$ for $E_1^{i} + E_2^{j}=X$. 

Say the set of energies $\mathcal{E}_{12}=\mathcal{E}_{1}+\mathcal{E}_{2}$ in which the baths $B_1$ and $B_2$ jointly live with high probability and satisfy the properties mentioned for Eq.~\eqref{eq:BathTypicality}. Then, for $E_1 \in \mathcal{E}_1$ and $E_2 \in \mathcal{E}_2$, the normalized joint state $\rho^{E_1 + E_2}_{B_1B_2S_{12}}$ after the projection with the projector $P_{E_1 + E_2}$ is
\begin{align}
	\rho_{E_1 + E_2}^{B_1B_2S_{12}} & = \frac{1}{p_{E_1 + E_2}} P_{E_1 + E_2} \left( \gamma_{B_1} \otimes \gamma_{B_2} \otimes \rho_{S_{12}} \right) P_{E_1 + E_2}, \nonumber \\
	& \approx \bigoplus_{E_{S_1} + E_{S_2}} \eta_{E_1-E_{S_1} + E_2-E_{S_2}}^{B_1B_2} \otimes P_{E_{S_1} + E_{S_2}} (\rho_{S_{12}}) P_{E_{S_1} + E_{S_2}},
\end{align}
where $P_{E_{S_1} + E_{S_2}}$ are the projectors on the system ($S_{12}$) space spanning over the block given by the energy $E_{S_1} + E_{S_2}$, and 
\begin{align}
	\eta_{E_1-E_{S_1} + E_2-E_{S_2}}^{B_1B_2} = \eta_{E_1-E_{S_1}}^{B_1} \otimes \eta_{E_2-E_{S_2}}^{B_2}= \frac{\mathbb{I}_{E_1-E_{S_1}}^{B_1}}{g_{B_1}(E_1) e^{-\beta_1 E_{S_1}}} \otimes \frac{\mathbb{I}_{E_2-E_{S_2}}^{B_2}}{g_{B_2}(E_2) e^{-\beta_2 E_{S_2}}}.
\end{align}   
With these minimum implements in hand, we move on to characterize semi-local thermal operations below. 

\section{Characterization of semi-local thermal operations \label{sec:Character-slto}}

Recall the Stinespring dilation of the SLTOs given in the main text, 
\begin{align}
	\label{eq:slto-Stinespring}
	\Lambda_{S_{12}} \left(\rho_{S_{12}} \right)=\tr_{B_1 B_2} \left[ U (\gamma_{B_1} \otimes \gamma_{B_2} \otimes \rho_{S_{12}}) U^\dag \right],
\end{align} 
where the global unitary $U$ satisfies the commutation relations (R1) $\left[U, \ H_{B_1} + H_{B_1} + H_{B_2} + H_{S_2}  \right] =0$, and (R2) $\left[U, \ \beta_1 \ (H_{S_1} + H_{B_1}) + \beta_2 \ (H_{S_1} + H_{B_1})  \right] =0$. The $H_{S_x}$ and $H_{B_x}$ are the Hamiltonians of the subsystem $S_x$ and the bath $B_x$ respectively for $x=1,2$. The thermal states of the baths are denoted by $\gamma_{B_x}=\frac{e^{-\beta_xH_{B_x}}}{\tr [e^{-\beta_xH_{B_x}}]}$.

The unitary $U$ strictly conserves both the total energy and the total weighted-energy, as guaranteed by the commutation relations (R1) and (R2). It can be expressed in terms of the total energy blocks, as 
\begin{align}\label{eq:UniJoint}
	U= \sum_{E_1 + E_2} P_{E_1 + E_2} \left( U \right) P_{E_1 + E_2}=\bigoplus_{E_1 + E_2} U_{E_1 + E_2},
\end{align}
where $P_{E_1 + E_2}$ is the projector spanning the sub-space with total energy $E_1 + E_2$. The global unitary $U$ cannot transfer populations among the blocks with different total energies. The $U_{E_{1}+E_{2}}$ cannot be arbitrary. Within a block given by the fixed total energy $E_1 + E_2$, the $U_{E_1 + E_2}$ can implement population transfer among the states that have identical total weighted-energy $\beta_1 E_1 + \beta_2 E_2$. 

For a system state $\rho_{S_{12}}$, the $U_{E_1 + E_2}$ operates on the total energy block given by the normalized state system-baths composite $\rho_{E_1 + E_2}^{B_1B_2S_{12}}$, defined in the earlier section, that is
\begin{align}
	\rho_{E_1 + E_2}^{B_1B_2S_{12}} & = \frac{1}{p_{E_1 + E_2}} P_{E_1 + E_2} \left( \gamma_{B_1} \otimes \gamma_{B_2} \otimes \rho_{S_{12}} \right) P_{E_1 + E_2}, \nonumber \\
	& \approx \bigoplus_{E_1 + E_2} \frac{\mathbb{I}_{E_1-E_{S_1}}^{B_1}}{g_{B_1}(E_1) e^{-\beta_1 E_{S_1}}} \otimes \frac{\mathbb{I}_{E_2-E_{S_2}}^{B_2}}{g_{B_2}(E_2) e^{-\beta_2 E_{S_2}}} \otimes P_{E_{S_1} + E_{S_2}} (\rho_{S_{12}}) P_{E_{S_1} + E_{S_2}},
\end{align}
as discussed in Section \ref{sec:2bath2sys}. The $E_1$ and $E_2$ are the energies correspond to the composites $B_1S_1$ and $B_2S_2$. Recall, there are many combinations of $E_1$ and $E_2$ that led to same value of total energy $E_1 + E_2$ and total weighted-energy $\beta_1 E_1+ \beta_2 E_2$. 

\subsection{SLTOs are those that preserve semi-Gibbs states, and vice versa}
From the definition of SLTOs itself, it is clear that these operations preserve semi-Gibbs states. Let us consider the microscopic picture. Say the initial system state is in the semi-Gibbs state, given by
\begin{align}
	\rho_{S_{12}}=\gamma_{S_1} \otimes \gamma_{S_2}=\sum_i\frac{e^{-\beta_1 E^{S_1}_i}}{Z_{S_1}} \proj{E^{S_1}_i} \otimes \sum_j\frac{e^{-\beta_2 E^{S_2}_j}}{Z_{S_2}} \proj{E^{S_2}_j},
\end{align}
where $\ket{E^{S_{1}}_{i}}$ are the energy eigenstates of the Hamiltonian $H_{S_{1}}=\sum_{i} E^{S_{1}}_{i} \proj{E^{S_{1}}_{i}}$ of the subsystem $S_{1}$, and the $\ket{E^{S_{2}}_{j}}$ are the energy eigenstates of the Hamiltonian $H_{S_{2}}=\sum_{j} E^{S_{2}}_{j} \proj{E^{S_{2}}_{j}}$ of the subsystem $S_{2}$. The $Z_{S_{1}}$ and $Z_{S_{2}}$ are the partition functions. For each block with the total energy $E_1 + E_2$, the bath-system composite state, as discussed in Section \ref{sec:2bath2sys}, is 
\begin{align}
	\rho^{B_{1}B_{2}S_{12}}_{E_1 + E_2} &=  \bigoplus_{E_i^{S_1}+E_j^{S_2}} \frac{\mathbb{I}}{g_{B_1}(E_1)e^{-\beta_1 E_i^{S_1}}} \otimes \frac{\mathbb{I}}{g_{B_2}(E_2)e^{-\beta_2 E_j^{S_2}} } \otimes \left(\frac{e^{-\beta_1 E_i^{S_1}}}{Z_{S_1}} \proj{E^{S_1}_i} \otimes  \frac{e^{-\beta_2 E_j^{S_2}}}{Z_{S_2}} \proj{E^{S_2}_j}\right), \\
	& = \frac{\mathbb{I}}{g_{B_1}(E_1)Z_{S_1}} \otimes \frac{\mathbb{I}}{g_{B_2}(E_2)Z_{S_2} } \label{eq:MaxMixedSysBath}.
\end{align}
Clearly, the application of a unitary (with the form \eqref{eq:UniJoint}) on the joint system-bath composite that strictly conserves both the total energy and the total weighted-energy will not change the maximally mixed state \eqref{eq:MaxMixedSysBath} in every total energy block. Therefore, a semi-Gibbs state will not change upon the application of an SLTO. 

Now we consider the reverse statement that is if a semi-Gibbs state preserving operation can be implemented using an SLTO, placed in the Corollary below.

\begin{cor}[Semi-Gibbs preservation of semi-local thermal operations]\label{defi:slto2}
	Consider two non-interacting sub-systems $S_1$ and $S_2$, of a bipartite system $S_{12}$, that are semi-locally interacting with the baths $B_1$ and $B_2$, at inverse temperatures $\beta_1$ and $\beta_2$, respectively. Then the semi-local thermal operations are the ones that satisfy the semi-Gibbs preservation condition, 
	\begin{align}
		\Lambda_{S_{12}} \left( \gamma_{S_1} \otimes \gamma_{S_2} \right) = \gamma_{S_1}^\prime \otimes \gamma_{S_2}^\prime \in \mathcal{T}_{S_{12}}, \ \ \ \forall \ \gamma_{S_1} \otimes \gamma_{S_2} \in \mathcal{T}_{S_{12}}.
	\end{align}
\end{cor}

\begin{proof}
	Let us just consider the situation where the Hamiltonians of the sub-systems do not change, i.e., the situation where $\Lambda_{S_{12}} \left( \gamma_{S_1} \otimes \gamma_{S_2} \right)= \gamma_{S_1} \otimes \gamma_{S_2}$. Extension to the general cases can be simply followed. 
	
	Below we show that the semi-Gibbs preserving operations are precisely the semi-local thermal operations when they are applied on the block-diagonal states, i.e., $[\rho_{S_{12}},\ H_{S_{12}}]=0$. Let us consider that the Hamiltonian of the system $S_{12}$, 
	\begin{align}\label{Appeq:WeightedHamilS12}
		H_{S_{12}}=H_{S_1}+ H_{S_2},
	\end{align}
	is non-degenerate in system energy. The bath-system composite ($B_1B_2S_{12}$) in the total energy blocks $E_1 + E_2$ is expressed in Eq.~\eqref{eq:WeightedEngBlocksState}. We restrict ourselves within one total energy block $E_1 + E_2$. We show that, with the help of just permutation between bases, any arbitrary operation can be performed that preserves the corresponding \emph{semi-Gibbs} state within the total energy block. Once done with this, it will be easy to check that such an operation can be implemented in every other total energy blocks.  Note, the total energy  $E_1 + E_2 \in \mathcal{E}_{12}$ (see Section \ref{sec:2bath2sys}). 
	
	Within a block of total energy $E_1 + E_2$, there are sub-blocks corresponding to the system energies $E^{S_1}_i + E^{S_2}_j$.
	The permutations among the eigenbases respecting strict conservation of   total energy and total weighted-energy will result in transfer of eigenbases among the sub-blocks defined by $E^{S_1}_i + E^{S_2}_j$ (the energy of the system). Each sub-block also constitutes a degenerate subspace with the dimension $g_{B_1}(E_1) e^{-\beta_1 E^{S_1}_i} g_{B_2}(E_2) e^{-\beta_2 E^{S_2}_j}$ (see Section \ref{sec:2bath2sys}). All the eigenvalues of the system-bath composite in this sub-block become equal to
	\begin{align}
		\frac{p(E^{S_1}_i + E^{S_2}_j)}{g_{B_1}(E_1) e^{-\beta_1 E^{S_1}_i} g_{B_2}(E_2) e^{-\beta_2 E^{S_2}_j}},
	\end{align}
	after normalization, where $ p(E^{S_1}_i + E^{S_2}_j)=\tr [P_{E^{S_1}_i + E^{S_2}_j} \ \rho_{S_{12}}]$. For the notational simplicity, let us denote $p(E^{S_1}_i + E^{S_2}_j) \rightarrow p_{ij}$, where $i$ (and $j$) stands for the energy levels $E^{S_1}_i$ in system $S_1$ (and $E^{S_2}_j$ in system $S_2$), and $ g_{B_1}(E_1) e^{-\beta_1 E^{S_1}_i} \rightarrow d_i$ and $g_{B_2}(E_2) e^{-\beta_2 E^{S_2}_j} \rightarrow d_j $. 
	
	Now, we can introduce permutations among the states that preserves total energy as well as the total weighted-energy. Then, the 'transition current' between system energy sub-blocks is denoted by $t_{ij \rightarrow mn}$, which is equal to the number of eigenstates that are transferred from the $ij$-th sub-block (corresponding to the system energy $E^{S_1}_i + E^{S_2}_j$) to $mn$-th sub-block (corresponding to the system energy $E^{S_1}_m + E^{S_2}_n$). The transition current satisfies  
	\begin{align}
		&\sum_{ij}t_{ij\rightarrow mn}=d_md_n, \label{eq:TransitionCurrent1} \\
		&\sum_{mn}t_{ij\rightarrow mn}=d_id_j. \label{eq:TransitionCurrent2}
	\end{align}
	The permutations will lead to a modification in the probability distribution in the system part $\{p_{ij} \} \rightarrow \{q_{mn} \}$. The new probability distribution can be written in terms of the transition currents, satisfying \eqref{eq:TransitionCurrent1} and \eqref{eq:TransitionCurrent2}, becomes
	\begin{align}
		q_{mn}=\sum_{ij}  t_{ij\rightarrow mn} \ \frac{p_{ij}}{d_id_j}=\sum_{ij}  s_{ij\rightarrow mn} \ p_{ij},
	\end{align}
	where $s_{ij\rightarrow mn}=\frac{t_{ij\rightarrow mn}}{d_id_j}$ is the probability of the transition $ij\rightarrow mn$. The transition matrix $\{s_{ij\rightarrow mn}\}$ transforms a normalized probability distribution to another normalized probability distribution, as it satisfies the stochastic condition $\sum_{mn}  s_{ij\rightarrow mn}=1, \ \ \forall ij$. Along with the relation $\frac{d_id_j}{d_md_n}=\frac{e^{-\beta_1 E_i^{S_1} -\beta_2 E_i^{S_2} }}{e^{-\beta_1 E_m^{S_1} -\beta_2 E_n^{S_2} }}$, the stochastic condition implies that all transformations satisfying the constraints \eqref{eq:TransitionCurrent1} and \eqref{eq:TransitionCurrent2} guarantee the preservation of the semi-Gibbs state $\gamma_{S_1} \otimes \gamma_{S_2}$. With this, we prove that in a given total energy block, all possible operations that strictly conserve both the total energy and the total weighted-energy are semi-Gibbs preserving operations. 
	
	Note, for a given arbitrary semi-Gibbs preserving transformation on the system $S_{12}$ in a block-diagonal state $\rho_{S_{12}}$, a permutation among the system energy sub-blocks within a fixed total energy and total weighted-energy block, can be performed to result in the desired transformation. The resultant operation in the block strictly conserves total energy as well as total weighted-energy. Further for every such block, there exists permutation operations that lead to the same transformation on the system $S_{12}$ part. The combination of all these individual transformations, that are performed in different blocks leads to the implementation of the semi-Gibbs preserving operation on the initial state of the system $S_{12}$.
	
	It is clear from above that all semi-Gibbs preserving transformations can be performed with the help of permutations within each block having fixed total energy and total weighted-energy. These operations are unitary and strictly conserve the total energy as well as the total weighted-energy. Hence, these are nothing but the semi-local thermal operations. Therefore, an arbitrary semi-Gibbs preserving operation applied on system in a block-diagonal state can be simulated using semi-local thermal operation. 
\end{proof}

We shall re-consider this semi-Gibbs preserving property to characterize the state transformations under SLTOs in Section \ref{sec:StateTrasNoCat}, and in the context of majorization in the Theorem \ref{thm:MajorCond}. 

\subsection{Catalytic semi-local thermal operations (cSLTOs)}
The SLTOs can be further generalized with an access to a bipartite catalyst $C_{12}$ composed of two non-interacting subsystems $C_1$ and $C_2$ and the Hamiltonian $H_{C_{12}}=H_{C_1} + H_{C_2}$. The $C_1$ is clubbed with the subsystem $S_1$ to form the composite $S_1C_1$. Similarly, the $C_2$ is clubbed with the $S_2$ to form $S_2C_2$. Then, the composites $S_1C_1$ and $S_2C_2$ interacts with the baths $B_1$ and $B_2$ via semi-local thermal operations.  Such operations are called \emph{catalytic semi-local thermal operations} (cSLTOs) that satisfy 
\begin{align}
	\Lambda_{S_{12}C_{12}}(\rho_{S_{12}} \otimes\rho_{C_{12}} ) \rightarrow \sigma_{S_{12}} \otimes \rho_{C_{12}},
\end{align} 
where $\rho_{C_{12}} $ is a state of the catalyst. Note, the catalyst remains unchanged before and after the process. These catalytic operations form a larger set of thermodynamically allowed operations compared to SLTOs and respect all the properties satisfied by the SLTOs. The cSLTOs are the allowed thermodynamic operation in a quantum heat engine and constitute the free operation for the resource theory developed in this article. Several useful properties of these operations are outlined in the Appendix (also see Supplemental Information). It is interesting to note that the cSLTOs converge to the (local) thermal operations that are introduced in the resource theory of quantum states beyond thermal equilibrium presented in \cite{Brandao13, Horodecki13, Brandao15}, when both the baths are of the same temperature, i.e., for $\beta_1=\beta_2$. 

\subsection{SLTOs are time-translation symmetric operations}
It is interesting to note that the SLTOs, as well as the cSLTOs, are time-translation symmetric operations with respect to the time evolution generated by the Hamiltonian $H_{S_{12}}=H_{S_1} + H_{S_2}$. In other words, if there is a time translation of the system driven by unitary $V(t)=e^{-itH_{S_{12}}/\hbar}$ for any given time $t$, then 
\begin{align}\label{eq:ttsymtotalHam12}
	V(t) \left[\Lambda_{S_{12}}(\rho_{S_{12}})\right]V^\dag(t)= \Lambda_{S_{12}}\left[V(t) (\rho_{S_{12}})V^\dag(t)\right],
\end{align}
i.e., the order of the time translation operations and the SLTO commute. Because of this symmetric property, the SLTOs monotonically decrease the superpositions between different eigenbases of $H_{S_{12}}$ or, in other words, among the energy eigenbases. Note the SLTOs are also time-translation symmetric with respect to the system's weighted-Hamiltonian $H^{\beta_1 \beta_2}_{S_{12}}=\beta_1 H_{S_1} + \beta_2 H_{S_2}$, since $[H_{S_{12}}^{\beta_1 \beta_2}, H_{S_{12}}]=0$. 

The Eq.~\eqref{eq:ttsymtotalHam12} further implies that the SLTOs commute with the dephasing operations in the eigenbases of the Hamiltonian of the system $H_{S_{12}}=\sum_{i,j} (E^{S_1}_i + E^{S_2}_j) \proj{ij}$, i.e., 
\begin{align}
	\Lambda_{S_{12}}\circ P_{S_{12}} (\rho_{S_{12}}) = P_{S_{12}} \circ \Lambda_{S_{12}} (\rho_{S_{12}}), \label{eq:dephaseGlob}
\end{align}
where $P_{S_{12}}(\rho_{S_{12}})= \sum_{ij} \braket{ij|\rho_{S_{12}}|ij} \ketbra{ij}{ij}$ is the dephasing operation. Note the dephasing operation can be achieved by time averaging time-translated state for a long enough time $T$,
\begin{align}
	P_{S_{12}}(\rho_{S_{12}})=\frac{1}{T} \int_0^T V(t) \ (\rho_{S_{12}}) \ V(t) dt. \label{eqApp:DephasedStateWeightedHamil}
\end{align}
The Eq. \eqref{eq:dephaseGlob} signifies that the diagonal elements in the eigenbases of $H_{S_{12}}$, i.e., $P_{S_{12}}(\rho_{S_{12}})$, evolve independently of the off-diagonal elements. Further, the amount of asymmetry present in a state due to the superposition between different energy eigenbases monotonically decreases under the SLTOs. We use these properties to understand the free-entropy locking in superpositions and to add more conditions to supplement the second laws for state transformations below.

\section{Information theoretic notations and technical tools \label{sec:tools}}
In this section, we shall briefly outline the notations and tools that will be used to derive the conditions of state transformation under SLTOs. The interested readers are referred to \cite{Horodecki13, Brandao15} for more details.

\subsection{R\'enyi $\alpha$-entropies}
Given an $k$-dimensional probability distribution $p=\{p_i \}_{i=1}^k$,  the R\'enyi $\alpha$-entropies are defined as 
\begin{align}
	H_{\alpha} (p)=\frac{\mbox{sgn}(\alpha)}{1-\alpha} \log \sum_{i=1}^k p_i^\alpha, \ \ \forall \alpha \in \mathcal{R}\setminus \{0,1 \},
\end{align}
where $\mbox{sgn}(\alpha) =1$ for $\alpha \geqslant 0$ and $\mbox{sgn}(\alpha) = - 1$ for $\alpha < 0$. For $\alpha \in \{- \infty, 0, 1, \infty  \}$, the $H_{\alpha}$s can be computed using limits, and they are
\begin{align}
	H_{- \infty}(p)=\log p_{\min}, \ \ \ H_{0}(p)=\log \mbox{rank}(p), \ \ \ H_{1}(p)=-\sum_{i=1}^k p_i \log p_i, \ \ \mbox{and} \ \ H_{ \infty}(p)=- \log p_{\max}. 
\end{align}
The $\mbox{rank}(p) $ means the number of non-zero elements in $p$, and $p_{\min}$ and $p_{\max}$ are the elements with smallest and largest values in $p$. 

These R\'enyi entropies can also be defined for arbitrary quantum state $\rho$, where the $\{p_i \}$ would be the eigenvalues of the density matrix $\rho$. Note, at $\alpha \rightarrow 1$, the $ H_{1}(\rho)=-\tr \rho \log \rho$ becomes the von Neumann entropy of the state $\rho$. 

\subsection{R\'enyi $\alpha$-relative entropies}
For any two $k$-dimensional probability distributions $p=\{p_i \}_{i=1}^k$ and $q=\{q_i \}_{i=1}^k$, the R\'enyi $\alpha$-relative entropies are defined as
\begin{align}
	D_{\alpha}(p\parallel q)= \frac{\mbox{sgn}(\alpha)}{\alpha -1} \log \sum_{i=1}^k p_i^\alpha q_i^{1-\alpha},  \ \ \ \forall \alpha \in [-\infty, \infty].
\end{align}
For the cases $\alpha \in \{- \infty, 0, 1, \infty \} $, the $D_{\alpha}$s are calculated using limits, as
\begin{align}
	& D_{\infty}(p\parallel q)= \lim_{\alpha \rightarrow \infty} D_{\alpha}(p\parallel q)= \log \max_{i} \frac{p_i}{q_i}, \\
	& D_{ - \infty}(p\parallel q)= \lim_{\alpha \rightarrow -\infty} D_{\alpha}(p\parallel q) = D_{\infty}(q\parallel p), \\
	& D_{ 0}(p\parallel q)= \lim_{\alpha \rightarrow 0+} D_{\alpha}(p\parallel q) = -\log \sum_{i: p_i\neq 0}^k q_i, \\
	& D_{1}(p\parallel q)= \lim_{\alpha \rightarrow 1} D_{\alpha}(p\parallel q) = \sum_i p_i (\log p_i - \log q_i).
\end{align}
Here we use the conventions that $\frac{0}{0}=0$ and $\frac{x}{0}=\infty$ for $x>0$. The R\'enyi $\alpha$-relative entropies satisfy many interesting properties, and we shall mention few useful ones below. These entropies monotonically decrease under stochastic maps $\Lambda$, i.e.,
\begin{align}
	D_{\alpha}(p\parallel q) \geqslant D_{\alpha}(\Lambda(p) \parallel\Lambda(q)), \ \ \ \forall \alpha \in [-\infty, \infty].
\end{align}
The inequalities are also known as the data-processing inequality. Another important property is that, for  $\alpha \in [0, \infty] $,
\begin{align}
	D_{\alpha}(p\parallel q) \leqslant D_{\delta}(p\parallel q), \ \ \mbox{for} \ \ \alpha \leqslant \delta.
\end{align}

For a $k$-dimensional probability distribution $\{q_i\}_{i=1}^k$ with $0<q_i<1$ and $ \forall q_i \in \mathbb{Q}$, there exist a set of natural numbers $\{d_i\}_{i=1}^k$ such that $\sum_i d_i=N$ and $q_i=\frac{d_i}{N}$. Then a fine-grained, $N$-dimensional uniform probability can be written as
\begin{align}
	\Gamma(q)=\left\{ \underbrace{\frac{q_1}{d_1},\ldots,\frac{q_1}{d_1}}_{d_1},\ldots,\underbrace{\frac{q_k}{d_k},\ldots,\frac{q_k}{d_k}}_{d_k}   \right\} =\left\{\underbrace{\frac{1}{N},\ldots,\frac{1}{N}}_{N} \right\}. 
\end{align}
Similarly, any other probability distribution $\{ p_i\}_{i=1}^k$ can be fine grained to
\begin{align}
	\Gamma(p)=\left\{ \underbrace{\frac{p_1}{d_1},\ldots,\frac{p_1}{d_1}}_{d_1},\ldots,\underbrace{\frac{p_k}{d_k},\ldots,\frac{p_k}{d_k}}_{d_k}   \right\}.
\end{align}
Then, for $\alpha \in [-\infty, \infty]$, the R\'enyi $\alpha$-relative entropies are related to the R\'enyi $\alpha$-entropies as
\begin{align}
	D_\alpha \left(p|q \right)=D_\alpha\left(\Gamma(p)| \Gamma(q) \right)=\mbox{sgn}(\alpha)\log N -H_\alpha (p).
\end{align}
For the situations where $q_i \notin \mathbb{Q}$, we can relate the R\'enyi $\alpha$-relative entropies with the R\'enyi $\alpha$-entropies using following Lemma.
\begin{lem}[Ref. \cite{Brandao15}]
	Consider a non-increasingly ordered, $k$-dimensional probability distribution $q=\{q_i \}_{i=1}^k$ with $\forall q_i>0$, and $q_i$s may possibly assume irrational values. Then, for any $\epsilon >0$, there exists a probability distribution $q_\epsilon$ such that
	\item (i) $\parallel q-q_\epsilon \parallel < \epsilon$,
	\item (ii) Each probability in $q_\epsilon $ is rational so that $q_\epsilon =\left\{\frac{d_i}{N}  \right\}_{i=1}^k $, where $\forall d_i \in \mathbb{N}$ and $\sum_{i=1}^k d_i =N$.
	\item (iii) There exists a stochastic channel $\Lambda$ such that $\Lambda(q)=q_\epsilon$, and for any arbitrary probability distribution $r$, the channel satisfies $\parallel r- \Lambda(r) \parallel\leqslant O\left(\sqrt{\epsilon}\right)$.  
\end{lem}

The R\'enyi $\alpha$-relative entropies can be extended to two arbitrary quantum states $\rho$ and $\sigma$. For this work, we shall restrict to the cases where $[\rho, \sigma]=0$, and $\mbox{supp}[\rho] \subseteq \mbox{supp}[\sigma]$. Then the R\'enyi $\alpha$-relative entropies are defined as 
\begin{align}
	D_{\alpha}(\rho \parallel \sigma)=\frac{\mbox{sgn}(\alpha)}{\alpha -1} \log \tr [\rho^\alpha \sigma^{1-\alpha}], \ \ \ \forall \alpha \in [-\infty, \infty].
\end{align}
For $\alpha\rightarrow0$, it becomes the min-relative entropy, 
\begin{align}
	D_{\min}(\rho \parallel \sigma)=D_{0}(\rho \parallel \sigma)=-\log \tr[\Pi_{\rho}\sigma],                                                    
\end{align}
where $\Pi_{\rho} $ is the projector onto the support of the state $\rho$. For the $\alpha \rightarrow 1$, it reduces to the von Neumann relative entropy as
\begin{align}
	D_{1}(\rho \parallel \sigma)= \tr[\rho(\log\rho- \log\sigma].                                                    
\end{align}
For the case $ \alpha\rightarrow \infty$, it results in the max-entropy given by
\begin{align}
	D_{\max}(\rho \parallel \sigma)=D_{\infty}(\rho \parallel \sigma)=\log \min\{\lambda: \rho \leqslant \lambda \sigma \}.                                                    
\end{align}
The R\'enyi $\alpha$-relative entropies are known to satisfy the monotonicity relation under completely positive maps, for $\alpha \in [0, 2]$,
\begin{align}
	D_{\alpha}(\rho\parallel \sigma) \geqslant D_{\alpha}(\Lambda(\rho) \parallel\Lambda(\sigma)).
\end{align}
For other values of $\alpha$, validity of the monotonicity is still an open question.

\subsection{Majorization and catalytic majorization (trumping) \label{sec:MajorizationSchur}}
The majorization relations are useful to introduce partial orders between arbitrary probability distributions \cite{Marshall11}. For any two probability distributions $p=\{ p_i \}_{i=1}^k$ and $p^\prime=\{ p_i^\prime \}_{i=1}^k$, we say that $p$ majorizes $p^\prime$, i.e., $p \succcurlyeq  p^\prime$, if for all $l=1, \ldots, k$,
\begin{align}\label{eq:majorization}
	\sum_{i=1}^l p_i^{\downarrow} \geqslant \sum_{i=1}^l p_i^{\prime \downarrow} \ \ \ \mbox{and} \ \ \ \sum_{i=1}^k p_i^{\downarrow} = \sum_{i=1}^k p_i^{\prime \downarrow}=1,
\end{align}
where the $p^{\downarrow}$ is obtained by rearranging $p$ in the non-increasing order so that $p_1^{\downarrow}\geqslant p_2^{\downarrow}\geqslant \ldots \geqslant p_k^{\downarrow} $, and similarly we obtain $p^{\prime \downarrow}$ by rearranging $p^\prime$. When two probability distributions are partial ordered through a majorization relation, these satisfy the following properties: 

\noindent (i) Two probability distributions $p$ and $p^\prime$ satisfy a majorization relation, $p \succcurlyeq  p^\prime$ if and only if there exist a channel $\Lambda$ such that $p^\prime=\Lambda(p)$ and $\Lambda$ satisfies $\Lambda(\eta) =\eta$, where $\eta$ is the uniform distribution. The channels $\Lambda$ are called bi-stochastic channels, and these can be implemented using random unitary operations.

\noindent (ii) If two probability distributions $p$ and $p^\prime$ satisfy a majorization relation, $p \succcurlyeq  p^\prime$, then
\begin{align}
	f(p) \leqslant f(p^\prime),
\end{align}
where $f$ are all Schur-concave functions. Note, the R\'enyi $\alpha$-entropies $H_{\alpha}$s are Schur-concave functions for $\alpha \in [-\infty, \infty]$.

Let us now discuss catalytic majorization or trumping. There are situations where $p$ and $p^\prime$ cannot be partially ordered in terms of majorization, but in presence of an additional probability distribution $x$, it satisfies $p \otimes x \succcurlyeq  p^\prime \otimes x $. This is termed as ``$p$ trumps $p^\prime$'' or $p \succcurlyeq_T  p^\prime $. For two given probability distribution it is often very difficult to find the additional probability distribution $x$ to check if the former are related through trumping. However, if the two probability distributions satisfy following two Lemmas, then one could ensure the existence of at least one $x$.
\begin{lem}[Ref. \cite{Brandao15}]
	Let us consider two probability distributions $p$ and $p^\prime$ that do not contain any element equal to zero. Then,  $p$ trumps $p^\prime$, i.e., $p \succcurlyeq_T  p^\prime $ if, and only if, the R\'enyi $\alpha$-entropies satisfy
	\begin{align}
		H_{\alpha}(p) \leqslant H_{\alpha}(p^\prime), \ \ \ \forall \alpha \in (-\infty, \infty).
	\end{align}
\end{lem}
Obviously, there are situations where the $p$ and $p^\prime$ are not of full-ranks (i.e., with all non-zero elements). In this situation, the Lemma below holds.
\begin{lem}[Ref. \cite{Brandao15}]
	Let us consider two arbitrary probability distributions $p$ and $p^\prime$. Then the following are equivalent: 
	
	\noindent (i) For an arbitrary $\epsilon >0$, there exists a full rank probability distribution $p^\prime_{\epsilon}$ such that $ \parallel p^\prime -  p^\prime_{\epsilon} \parallel \leqslant \epsilon$ and $p$ trumps $p^\prime_{\epsilon}$ (i.e., $p \succcurlyeq_T p^\prime_{\epsilon}$).
	\noindent (ii) The inequalities are satisfied, $ H_{\alpha}(p) \leqslant H_{\alpha}(p^\prime), \ \ \ \forall \alpha \in (-\infty, \infty)$. 
\end{lem}

The notion of majorization and trumping can also be extended to quantum states, say between $\rho$ and $\rho^\prime$. Then, the majorization relation $\rho \succcurlyeq  \rho^\prime$ implies the relations \eqref{eq:majorization} where the $p^{\downarrow}$ and $p^{\prime \downarrow}$ are the non-increasingly ordered eigenvalues of $\rho$ and $ \rho^\prime $ respectively. 

\subsection{$d$-majorization and catalytic $d$-majorization}
Not only for two probability distributions, but the majorization-like partial ordering can also be drawn between two pairs of probability distributions. Consider two pairs of probability distributions $(p,q)$ and $(p^\prime, q^{\prime})$. The $(p,q)$ $d$-majorizes $(p^\prime, q^{\prime})$ if and only if
\begin{align}
	\sum_i q_i f\left(\frac{p_i}{q_i}\right) \leqslant \sum_i q_i^\prime f\left(\frac{p_i^\prime}{q_i^\prime}\right),
\end{align}
for any arbitrary concave function $f$. This $d$-majorization based pre-ordering is then denoted as $d(p|q) \succcurlyeq d(p^\prime | q^{\prime})$. Given this definition of $d$-majorization, we present the following Lemma.

\begin{lem}[Ref. \cite{Brandao15}]
	Let us consider four probability distributions $p, p^\prime,  q$, and $ q^{\prime}$. Then the following statements are equivalent: 
	
	\noindent (i) The pair $(p,q)$ $d$-majorizes the pair $(p^\prime,q^\prime)$, i.e., $d(p|q) \succcurlyeq d(p^\prime | q^{\prime}) $. (ii) There exists a stochastic channel $\Lambda$ that satisfies $\Lambda(p)=p^\prime$ and $\Lambda(q)=q^\prime$.
\end{lem}

A catalytic $d$-majorization can also be introduced as in the following.
\begin{lem}[Ref. \cite{Brandao15}] \label{lem:CatDmajor}
	For two pairs of probability distributions, $(p,q)$ and $(p^\prime, q^{\prime})$ with the constraints that $q$ and $q^\prime$ are of full rank, the following conditions are equivalent:
	\item (i) The R\'enyi $\alpha$-relative entropies satisfy  $D_\alpha(p \parallel q) \geqslant D_\alpha(p^\prime \parallel q^\prime)$, $\forall \alpha \in [-\infty, \infty]$.
	\item (ii) For $\epsilon >0$, there exists full-rank probability distributions $r$, $s$, and $p_\epsilon^\prime$, and a stochastic channel $\Lambda$ such that
	\subitem (a) $\Lambda (p \otimes r)=p_\epsilon^\prime \otimes r$,
	\subitem (b) $\Lambda (q \otimes s)=q^\prime \otimes s$, moreover $s$ can be a uniform distribution $\eta$ onto the support of $r$,
	\subitem (c) $\parallel p^\prime - p_\epsilon^\prime \parallel \ \leqslant \epsilon$.
\end{lem}

We shall use this Lemma for the derivations of the second laws for the state transformations under cSLTOs.

\section{Second laws for transformations between states block-diagonal in energy eigenbases \label{sec:2ndLawsDiagStates}}
In this section, we present the necessary and sufficient conditions for state transformations, i.e., the second laws of state transformations, under (catalytic) semi-local thermal operations. All the necessary and sufficient conditions derived here are based on the assumption that the initial states are block-diagonal in the energy eigenbases of the system. However, the conditions still apply to initial non-block-diagonal states, but only as the necessary conditions. 

We start with a bipartite system $S_{12}$, with non-interacting Hamiltonian $H_{S_{12}}=H_{S_1} + H_{S_2}$, and the initial state $\rho_{S_{12}}$. The subsystem $S_1$ ($S_2$) is semi-locally interacting with a bath $B_1$ ($B_2$) at inverse temperature $\beta_1$ ($\beta_2$). After a transformation, the final state of the system becomes $\sigma_{S_{12}}$. At this stage, we assume that the system Hamiltonian remain unchanged before and after the transformation. Our goal is to find necessary and sufficient conditions with which we can ascertain that the transformation  
\begin{align}
	(\rho_{S_{12}}, H_{S_{12}}) \rightarrow (\sigma_{S_{12}}, H_{S_{12}})
\end{align} 
is possible via a semi-local thermal operation $\Lambda$, and vice versa.  

\subsection{State transformation in absence of a catalyst \label{sec:StateTrasNoCat}}
Suppose, we do not have access to a catalyst. Following the discussion made in Sections \ref{sec:2bath2sys} and \ref{sec:Character-slto}, we rewrite the initial system-bath composite as   
\begin{align}
	\gamma_{B_1} \otimes \gamma_{B_2} \otimes \rho_{S_{12}} =\sum_{E_{12}} P_{ E_{12}} \left( \gamma_{B_1} \otimes \gamma_{B_2} \otimes \rho_{S_{12}}  \right) P_{E_{12}} 
	=\sum_{E_{12}} p_{E_{12}} \ \rho^{B_{1}B_2S_{12}}_{E_{12}}.
\end{align}
Here $E_{12}=E_1 + E_2$ the total energy and the projector $ P_{ E_{12}}$ spans the system-bath joint space with the same value of total energy $E_{12}$. The probabilities $p_{E_{12}}=\tr[P_{E_{12}} (\gamma_{B_1} \otimes \gamma_{B_2} \otimes \rho_{S_{12}})] $. Now for $E_1, E_{B_1} \in \mathcal{E}_1$ and $E_2, E_{B_2} \in \mathcal{E}_2$, the normalized state of the bath-system composites, in a block of total energy $E_{12}=E_1 + E_2$, is given by
\begin{align}\label{eq:InState1}
	\rho^{B_{1}B_2S_{12}}_{E_{12}}=\bigoplus_{E_{S_{12}}} \frac{\mathbb{I}}{g_{B_1}(E_1)e^{-\beta_1 E_{S_1}}} \otimes \frac{\mathbb{I}}{g_{B_2}(E_2)e^{-\beta_2 E_{S_2}} } \otimes P_{E_{S_{12}}} (\rho_{S_{12}}) P_{E_{S_{12}}}
	=\bigoplus_{E_{S_{12}}} \eta^{B_1}_{E_{1}-E_{S_{1}}} \otimes \eta^{B_2}_{E_{2}-E_{S_{2}}} \otimes P_{E_{S_{12}}} (\rho_{S_{12}}) P_{E_{S_{12}}},
\end{align}
where $ E_{S_{12}}=E_{S_1} + E_{S_2}$ is the system energy and $P_{E_{S_{12}}}$ is the projector correspond to the system energy $E_{S_{12}}$. Similarly, the final joint state of the bath and the system composite can be expressed as
\begin{align}
	\gamma_{B_1} \otimes \gamma_{B_2} \otimes \sigma_{S_{12}} =\sum_{E_{12}} P_{ E_{12}} \left( \gamma_{B_1} \otimes \gamma_{B_2} \otimes \sigma_{S_{12}}  \right) P_{E_{12}}
	=\sum_{E_{12}} q_{E_{12}} \ \sigma^{B_{1}B_2S_{12}}_{E_{12}},
\end{align}
where $q_{E_{12}}=\tr[P_{E_{12}} (\gamma_{B_1} \otimes \gamma_{B_2} \otimes \sigma_{S_{12}})]$. For a total energy $E_{12}=E_1 + E_2$ block, the normalized state becomes
\begin{align}\label{eq:InState2}
	\sigma^{B_{1}B_2S_{12}}_{E_{12}}=\bigoplus_{E_{S_{12}}} \eta^{B_1}_{E_{1}-E_{S_{1}}} \otimes \eta^{B_2}_{E_{2}-E_{S_{2}}} \otimes P_{E_{S_{12}}} (\sigma_{S_{12}}) P_{E_{S_{12}}}.
\end{align}
With these structures of the initial and final states of the system-bath composites, we put forward the necessary and sufficient conditions for the transformations of block-diagonal states under semi-local thermal operations.  

\subsubsection{Majorization condition}
The conditions are derived in terms of majorization (see Section \ref{sec:MajorizationSchur}) in the following theorem. 

\begin{thm}[Majorization condition for state transformations]\label{thm:MajorCond}
	Consider two states $\rho_{S_{12}}$ and $\sigma_{S_{12}}$ that are block-diagonal in the eigenbases of the Hamiltonian of the system  $H_{S_{12}}=H_{S_1} + H_{S_2}$. Then the transformation $( \rho_{S_{12}}, H_{S_{12}}) \rightarrow (\sigma_{S_{12}}, H_{S_{12}})$ by means of semi-local thermal operation is possible if, and only if, the initial and final states of the system-bath composites satisfy the majorization relation 
	\begin{align}\label{eq:MajorCond1}
		\bigoplus_{E_{S_{1}}+E_{S_2}} \eta^{B_1}_{E_{1}-E_{S_{1}}} \otimes \eta^{B_2}_{E_{2}-E_{S_{2}}} \otimes P_{E_{S_{12}}} (\rho_{S_{12}}) P_{E_{S_{12}}} \succcurlyeq \bigoplus_{E_{S_{1}}+E_{S_2}} \eta^{B_1}_{E_{1}-E_{S_{1}}} \otimes \eta^{B_2}_{E_{2}-E_{S_{2}}} \otimes P_{E_{S_{12}}} (\sigma_{S_{12}}) P_{E_{S_{12}}} 
	\end{align}
	for large enough $E_1$ and $E_2$. 
	
	Moreover, for the cases where the initial system state is not block-diagonal in the eigenbases of $ H_{S_{12}}$, the necessary condition for the transformation $( \rho_{S_{12}}^\prime, H_{S_{12}}) \rightarrow (\sigma_{S_{12}}^\prime, H_{S_{12}})$ is
	\begin{align}\label{eq:MajorCond2}
		\bigoplus_{E_{S_{1}}+E_{S_2}} \eta^{B_1}_{E_{1}-E_{S_{1}}} \otimes \eta^{B_2}_{E_{2}-E_{S_{2}}} \otimes P_{E_{S_{12}}} (\rho_{S_{12}}^\prime) P_{E_{S_{12}}} \succcurlyeq \bigoplus_{E_{S_{1}}+E_{S_2}} \eta^{B_1}_{E_{1}-E_{S_{1}}} \otimes \eta^{B_2}_{E_{2}-E_{S_{2}}} \otimes P_{E_{S_{12}}} (\sigma_{S_{12}}^\prime) P_{E_{S_{12}}}.
	\end{align}
\end{thm}

\begin{proof}
	For the system-bath composite state in a block with fixed total energy $E_{12}=E_1+E_2$, the allowed operations are any unitary operations that also strictly conserve total weighted-energy $E_{12}^{\beta_1 \beta_2}=\beta_1 E_1 + \beta_2 E_2$. Because, such a unitary operation does not alter the total energy as well as the total weighted-energy of the block. However, as we shall show below, we not only can apply all such unitary operations but also implement random unitary operations as long as population remain in this fixed total energy and total weighted-energy block. Note these are the only possible operations that are allowed on the systems-baths joint space as they strictly conserve the total energy and total weighted-energy.
	
	We start with the first part of the theorem and prove it using a protocol involving the following steps, where the initial state of the system-bath composite is with total energy $E_{12}$ and total weighted-energy $E_{12}^{\beta_1 \beta_2}$:
	
	(i) {\it Implementing random unitary operations:} We assume that the $E_1$ and the $E_2$, corresponding to the block with energy $E_{12}$, are reasonably large. This in turn implies that the dimension of the maximally mixed state $\eta^{B_1}_{E_{1}-E_{S_{1}}} \otimes \eta^{B_2}_{E_{2}-E_{S_{2}}}$ is exponentially large. Therefore we can divide the state in two normalized sub-parts, as  
	\begin{align}
		\eta^{B_1}_{E_{1}-E_{S_{1}}} \otimes \eta^{B_2}_{E_{2}-E_{S_{2}}}= \left(\eta^{B_1}_{E_{1}^{\prime \prime}} \otimes \eta^{B_2}_{E_{2}^{\prime \prime}}\right) \otimes \left(\eta^{B_1}_{E_{1}^{\prime}-E_{S_{1}}} \otimes \eta^{B_2}_{E_{2}^{\prime}-E_{S_{2}}} \right)
	\end{align}
	where $E_1=E_1^{\prime \prime} + E_1^{\prime}$ and $E_2=E_2^{\prime \prime} + E_2^{\prime}$, and each of the sub-parts are in maximally mixed states with exponentially large dimensions. We further assume that $\eta^{B_1}_{E_{1}-E_{S_{1}}} \otimes \eta^{B_2}_{E_{2}-E_{S_{2}}}$ is so large that it hardly differ from   
	$\eta^{B_1}_{E_{1}^{\prime}-E_{S_{1}}} \otimes \eta^{B_2}_{E_{2}^{\prime}-E_{S_{2}}}$. The initial state of the system-bath composite with total energy $E_{12}$, is then
	\begin{align}
		\left(\eta^{B_1}_{E_{1}^{\prime \prime}} \otimes \eta^{B_2}_{E_{2}^{\prime \prime}}\right) \otimes \left(\bigoplus_{E_{S_{12}}} \eta^{B_1}_{E_{1}^{\prime}-E_{S_{1}}} \otimes \eta^{B_2}_{E_{2}^{\prime}-E_{S_{2}}} \otimes P_{E_{S_{12}}} (\rho_{S_{12}}) P_{E_{S_{12}}} \right).
	\end{align}
	Now, the state $\eta^{B_1}_{E_{1}^{\prime \prime}} \otimes \eta^{B_2}_{E_{2}^{\prime \prime}}$ can be used as a control state to implement random unitary operations (i.e., a unital channel) on the rest of the system-bath composite, by using a global unitary on the entire system-bath composite that strictly conserves both total energy and total weighted-energy. Lets say, we implement such a global unitary $U_{E_1+E_2}$ (see Eq.~\eqref{eq:UniJoint}) so that the resultant system state becomes $\sigma_{S_{12}}$ after tracing out the baths. Note the joint state of the system-bath composite can still have correlations among them.
	
	(ii) {\it Destroying unwanted correlations in the system-bath composite state:} Now, we can destroy the unnecessary correlations that may possibly present between the sub-systems and baths after step (i). That is done by a ``twirling" operation which is itself a random unitary operation within the total energy block. For each system state with the energy $E_{S_{12}}=E_{S_1}+E_{S_2}$, we apply twirling operation on the bath part $B_1B_2$ while applying identity operation on the system part, such that the transformed final system-bath state becomes classically correlated as
	\begin{align}
		\bigoplus_{E_{S_{1}}+E_{S_2}} \eta^{B_1}_{E_{1}-E_{S_{1}}} \otimes \eta^{B_2}_{E_{2}-E_{S_{2}}} \otimes P_{E_{S_{12}}} (\sigma_{S_{12}}) P_{E_{S_{12}}}.
	\end{align}
	
	Both the operations applied in steps (i) and (ii) are random unitary operations that respect strict conservation of total energy and total weighted-energy, and these are the precursors of semi-local thermal operations in the joint space of systems and baths (see Section \ref{sec:Character-slto}). Further, if two states are related through random unitary operations, then the states satisfy pre-ordering by a majorization relation, and this is a necessary and sufficient condition. Consequently, the transformation is possible if and only if the initial system-bath state majorizes the final one, within a block of total energy $E_{12}$. Now since baths are considerably large in energy compared to the systems, we can implement the random unitary operations in the other blocks, similar to the steps (i) and (ii), such that the reduced system state becomes exactly the same. As a result, we do not need to check the majorization relation for every block with fixed total total energy. Therefore, the transition $( \rho_{S_{12}}, H_{S_{12}}) \rightarrow (\sigma_{S_{12}}, H_{S_{12}})$ is possible if, and only if, the majorization relation \eqref{eq:MajorCond1} is satisfied. With this, we complete the proof of the first part.   
	
	For the second part, we recall that the reduced operations on the system part, as the result of global unitary operations on the systems-baths composite, respect time-translation symmetry with respect to time translation by the Hamiltonian of the system  $H_{S_{12}}$ (see main text). The block-diagonal elements (also known as the 'zero' mode elements) of the system density matrix ($\rho_{S_{12}}^\prime$) with respect to the eigenbases of $H_{S_{12}}$ evolve independently of the off-diagonal elements ('non-zero' modes) under this time-translation symmetric operations. For any transformation connecting the initial state ($\rho_{S_{12}}^\prime$) to a final state ($\sigma_{S_{12}}^\prime$), it is necessary that the corresponding block-diagonal states must satisfy the majorization relation \eqref{eq:MajorCond2}. This relation cannot provide the sufficient condition as it does not encode information related to the non-block-diagonal elements.
\end{proof}

So far, we have derived the necessary and sufficient condition for transformations between block-diagonal states under semi-local thermal operations. However, the condition requires us a take into account both the system and the bath parameters simultaneously, which is not always practical.

\subsubsection{Thermo-majorization condition}
Here we aim to derive necessary and sufficient conditions that exclusively depend on the system parameters, based on thermo-majorization. Consider two quantum states $\rho_{S_{12}}$ and $\sigma_{S_{12}}$, block-diagonal in the eigenbases of the Hamiltonian $H_{S_{12}}= H_{S_1} + H_{S_2}$, with the probabilities $\{ p_{ij} \}$ and $\{ q_{ij} \}$ respectively. Here the Hamiltonian of the system is written as $H_{12}=\sum_{ij} (E_i^{S_1}+E_j^{S_2}) \proj{ij}$, the probabilities are given by $p_{ij}=\bra{ij} \rho_{S_{12}} \ket{ij}$ and $q_{ij}=\bra{ij} \sigma_{S_{12}} \ket{ij}$. A pre-ordering is done by a non-increasing ordering of the quantities $\{p_{ij} \ e^{(\beta_1 E_i^{S_1} + \beta_2 E_j^{S_2})}\}$ and relabeled, so that
\begin{align}
	p_{11}  \ e^{(\beta_1 E_1^{S_1} + \beta_2 E_1^{S_2})} \geqslant p_{12}  \ e^{(\beta_1 E_1^{S_1} + \beta_2 E_2^{S_2})} \geqslant p_{21}  \ e^{(\beta_1 E_2^{S_1} + \beta_2 E_1^{S_2})} \geqslant p_{22}  \ e^{(\beta_1 E_2^{S_1} + \beta_2 E_2^{S_2})} \geqslant \ldots .
\end{align}
This determines the sequence in $\{p_{ij}\}$ which may or may not satisfy a non-decreasing order. We denote the set of the ordered probability distribution as $\{p_{ij}^\downarrow \}$, where $p_{11}^\downarrow $ is the $p_{ij}$ corresponding to the largest $p_{ij} \ e^{(\beta_1 E_i^{S_1} + \beta_2 E_j^{S_2})}$ value and so on. A similar pre-ordering is also done for $\{ q_{ij} \} \rightarrow \{ q_{ij}^\downarrow \}$. Now we construct a Lorentz curve with the points correspond to the pairs 
\begin{align*}
	\{(x,y)\}= &\left\{ (0, \ \ 0), \right. \\
	& \ (p_{11}, \ \ e^{-(\beta_1 E_1^{S_1} + \beta_2 E_1^{S_2})}), \\ 
	& \ (p_{11} + p_{12}, \ \ e^{-(\beta_1 E_1^{S_1} + \beta_2 E_1^{S_2})}+e^{-(\beta_1 E_1^{S_1} + \beta_2 E_2^{S_2})}), \\
	& \ (p_{11} + p_{12} + p_{21}, \ \ e^{-(\beta_1 E_1^{S_1} + \beta_2 E_1^{S_2})}+e^{-(\beta_1 E_1^{S_1} + \beta_2 E_2^{S_2})}+ e^{-(\beta_1 E_2^{S_1} + \beta_2 E_1^{S_2})}), \\
	& \ (p_{11} + p_{12} + p_{21} +p_{22}, \ \ e^{-(\beta_1 E_1^{S_1} + \beta_2 E_1^{S_2})}+e^{-(\beta_1 E_1^{S_1} + \beta_2 E_2^{S_2})}+ e^{-(\beta_1 E_2^{S_1} + \beta_2 E_1^{S_2})} + e^{-(\beta_1 E_2^{S_1} + \beta_2 E_2^{S_2})}), \\
	& \ \ \ \ \ \ \ \ \vdots \\
	& \ \left. (1, \ \ Z_1 Z_2) \right\}
\end{align*}
Plotting these points gives a function $f_p(x)$ corresponding to the state $\rho_{S_{12}}$. A similar function is also derived for $\{ q_{ij} \}$, and that is $f_q(x)$ for $\sigma_{S_{12}}$. 
\begin{thm}[Thermo-majorization condition for state transformations]\label{thm:ThermoMajorCond}
	A transition $( \rho_{S_{12}}, H_{S_{12}}) \rightarrow (\sigma_{S_{12}}, H_{S_{12}})$ can occur under semi-local thermal operation if, and only if, the spectra of $\rho_{S_{12}}$ thermo-majorizes the spectra of $\sigma_{S_{12}}$, i.e., 
	\begin{align}
		f_p(x) \geqslant f_q(x), \ \ \forall x \in [0, Z_1Z_2].
	\end{align}
\end{thm}

In this sub-section, the derived necessary and sufficient conditions based on majorization are very handy. This is in the sense that they are easy to check, in particular, the one based on thermo-majorization. However, as we have mentioned in Section \ref{sec:tools}, there are probability distribution that do not satisfy majorization relation as such, but can still possess a majorization based pre-ordering by having access to another probability distribution as a catalyst. It is not easy to check whether there exists a probability distribution which can act as a catalyst in order satisfy a majorization relation. We can still find necessary and sufficient condition(s) based on Re\'nyi relative entropies to ensure their existence, which we shall consider in the next sub-section.

\subsection{Catalytic state transformation}
Before we propose the necessary and sufficient conditions for the transformation under catalytic semi-local thermal operations, let us introduce the definition of $\alpha$-free-entropy ($S_{\alpha}$) in terms of the R\'enyi $\alpha$-relative entropy ($D_\alpha$).

\subsubsection{$\alpha$-free-entropies}

Consider a state $\rho_{S_{12}}$ of a bipartite system $S_{12}$ with Hamiltonian $H_{S_{12}}=H_{S_1} + H_{S_2}$, where the sub-system $S_{1}$ ($S_{2}$) is semi-locally interacting with the bath $B_1$ ($B_2$) at inverse temperature $\beta_1$ ($\beta_2$). Also, the state $\rho_{S_{12}}$ is block-diagonal in the eigenbases of the Hamiltonian $H_{S_{12}}$. Then, the R\'enyi  $\alpha$-relative entropy between the system state and its corresponding semi-Gibbs state, for $\alpha \in [-\infty, \infty]$, is given by
\begin{align}
	D_{\alpha} (\rho_{S_{12}} \parallel \gamma_{S_1} \otimes\gamma_{S_2}  )=\frac{\mbox{sgn}(\alpha)}{\alpha -1} \log \tr [(\rho_{S_{12}})^{\alpha} \ (\gamma_{S_1} \otimes\gamma_{S_2})^{1-\alpha}],
\end{align}
where the thermal states are $\gamma_{S_i}=\frac{e^{-\beta_i H_{S_i}}}{Z_i}$, with $ Z_i=\tr[e^{-\beta_i H_{S_i}}]$ and $i=1,2$. Now the $\alpha$-free-entropy of the state $\rho_{S_{12}}$ is defined as 
\begin{align}\label{eq:free-ent}
	S_\alpha (\rho_{S_{12}}, \gamma_{S_1}\otimes \gamma_{S_2})= D_\alpha \left(\rho_{S_{12}} \parallel \gamma_{S_1}\otimes \gamma_{S_2} \right) - \log Z_1Z_2, \ \ \forall \alpha \in [-\infty, \infty]. 
\end{align}
The name 'free-entropy' is justified by fact that it quantifies the work potential stored in a system in terms of entropy, which we shall discuss in the Section \ref{sec:fed}. For $\alpha \rightarrow 1$, the $S_\alpha$ reduces to the Helmholtz free-entropy as 
\begin{align}
	S_1 (\rho_{S_{12}}, \gamma_{S_1}\otimes \gamma_{S_2}) = \beta_1 E_{S_1} + \beta_2 E_{S_2}-S(\rho_{S_{12}}), 
\end{align}
where $S(\rho_{S_{12}})=-\tr \rho_{S_{12}} \ln \rho_{S_{12}}$ is the von Neumann entropy and $E_{S_x}=\tr [\rho_{S_{x}}H_{S_x}]$ is the average energy of the sub-system $S_x$, with $x=1,2$. 
For the cases where the state $\rho_{S_{12}}$ is uncorrelated, $\rho_{12}=\rho_{S_1} \otimes \rho_{S_2}$, the $\alpha$-free-entropy becomes additive $S_\alpha (\rho_{S_{12}}, \gamma_1\otimes \gamma_2)=S_\alpha (\rho_{S_{1}}, \gamma_1)+S_\alpha (\rho_{S_{2}},  \gamma_2)$.

\subsubsection{Second laws in terms of $\alpha$-free-entropies}
Now with the notion of $\alpha$-free-entropy, we go on to propose the necessary and sufficient conditions for catalytic semi-local thermal operation, in the following theorem.
\begin{thm}\label{thm:2ndLaws}
	Consider two states $\rho_{S_{12}}$ and  $\sigma_{S_{12}}$ that are block-diagonal in the eigenbases of a system Hamiltonian $H_{S_{12}}=H_{S_1} + H_{S_2}$. Then, a transformation $ (\rho_{S_{12}}, H_{S_{12}}) \rightarrow (\sigma_{S_{12}}, H_{S_{12}})$ is possible under semi-local thermal operation if, and only if, for all $\alpha \in (-\infty, \infty)$,
	\begin{align}\label{eq:Cat2ndLaws}
		S_\alpha(\rho_{S_{12}}, \gamma_{S_1} \otimes \gamma_{S_2}) \geqslant S_\alpha(\sigma_{S_{12}}, \gamma_{S_1} \otimes \gamma_{S_2}).
	\end{align}
\end{thm}

\noindent {\bf Remark:} {\it The condition for $\alpha=-\infty, \infty$ can be included by continuity.}

\begin{proof}
	Since the initial and final states are block-diagonal in the energy eigenbases, the theorem above can be proved using catalytic $d$-majorization shown in the Lemma \ref{lem:CatDmajor}. Replacing the probability distributions, in Lemma \ref{lem:CatDmajor}, with the eigenvalues of the states block-diagonal in energy, as
	\begin{align*}
		& p \rightarrow \rho_{S_{12}}; \ \ \mbox{initial state of the system}, \\
		& p^\prime \rightarrow \sigma_{S_{12}}; \ \ \mbox{final state of the system}, \\
		& q=q^\prime \rightarrow \gamma_{S_1} \otimes \gamma_{S_2}; \ \ \mbox{semi-Gibbs state of the system}, \\
		& r \rightarrow \rho_{C_{12}}; \ \ \mbox{a catalyst} \\
		& s \rightarrow \gamma_{C_1} \otimes \gamma_{C_2}; \ \ \mbox{semi-Gibbs state of the catalyst}.
	\end{align*}
	Note, if we consider that the Hamiltonian of the catalyst $H_{C_{12}}=H_{C_1} + H_{C_2}$ is trivial (i.e., $H_{C_{12}}=\mathbb{I}$), then  
	\begin{align*}
		s=\gamma_{C_1} \otimes \gamma_{C_2} \rightarrow \eta;  \ \ \mbox{a uniform distribution}.
	\end{align*}
	
	Let us first assume that the catalyst possesses a trivial Hamiltonian $H_{C_{12}}=\mathbb{I}$ and the conditions \eqref{eq:Cat2ndLaws} are satisfied. The latter, in terms of R\'enyi $\alpha$-relative entropy, means
	\begin{align}\label{eq:Cat2ndLawsD}
		D_\alpha(\rho_{S_{12}} \parallel \gamma_{S_1} \otimes \gamma_{S_2}) \geqslant D_\alpha(\sigma_{S_{12}} \parallel \gamma_{S_1} \otimes \gamma_{S_2}), \ \ \forall \alpha \in (-\infty, \infty).
	\end{align}
	Then, as the Lemma \ref{lem:CatDmajor} implies, there exists a catalyst $\rho_{C_{12}}$ and a channel $\Lambda$ that (i) preserves the semi-Gibbas state $\gamma_{S_1} \otimes \gamma_{S_2} \otimes \gamma_{C_1} \otimes \gamma_{C_2}$, as 
	\begin{align}
		\Lambda \left(\gamma_{S_1} \otimes \gamma_{S_2} \otimes \gamma_{C_1} \otimes \gamma_{C_2} \right) =\gamma_{S_1} \otimes \gamma_{S_2} \otimes \gamma_{C_1} \otimes \gamma_{C_2},
	\end{align}
	and (ii) transforms the initial state as
	\begin{align}
		\Lambda \left(\rho_{S_{12}} \otimes \rho_{C_{12}}  \right) = \rho_{S_{12}}^{o} \otimes \rho_{C_{12}},
	\end{align}
	where $\parallel \rho_{S_{12}}^{o} - \sigma_{S_{12}} \parallel \ \leqslant \epsilon$. As the operations that preserve semi-Gibbs states are also semi-local thermal operations (see Corollary \ref{defi:slto2}), it implies that such a transformation using catalytic semi-local thermal transformation is possible.
	
	Let us consider the converse now. Suppose there is a catalytic semi-local thermal channel $\Lambda$ that transforms 
	\begin{align}
		\Lambda \left(\rho_{S_{12}} \otimes \rho_{C_{12}}  \right) = \rho_{S_{12}}^{o} \otimes \rho_{C_{12}},
	\end{align}
	where $\parallel \rho_{S_{12}}^{o} - \sigma_{S_{12}} \parallel \ \leqslant \epsilon$. Then the Lemma \ref{lem:CatDmajor} implies that the conditions \eqref{eq:Cat2ndLawsD} are satisfied. This completes the proof.
\end{proof}

\subsubsection{For block-diagonal input state of a system, a block-diagonal catalyst is enough}
Note, for a block-diagonal input state of a system, a block-diagonal catalyst is enough. This can be seen from the fact that the semi-local thermal operations are time-translation symmetric with respect to the Hamiltonian of the systems $H_{S_{12}}$ (see main text). This is also true in the presence of catalysts. Therefore, the catalytic semi-local thermal operations are time-translation symmetric with respect to the joint Hamiltonian of the system and the catalyst, $H_{S_{12}} + H_{C_{12}}$. Mathematically, for the catalytic semi-local thermal transformation  $\rho_{S_{12}} \otimes \rho_{C_{12}} \rightarrow \Lambda(\rho_{S_{12}} \otimes \rho_{C_{12}})= \sigma_{S_{12}} \otimes \rho_{C_{12}}$, it means 
\begin{align}
	\Lambda \left(e^{- i t (H_{S_{12}} + H_{C_{12}})} \ \rho_{S_{12}} \otimes \rho_{C_{12}} \ e^{ i t (H_{S_{12}} + H_{C_{12}})} \right) =e^{- i t (H_{S_{12}} + H_{C_{12}})} \ \Lambda \left( \rho_{S_{12}}\otimes \rho_{C_{12}} \right) \ e^{ i t (H_{S_{12}} + H_{C_{12}})}.
\end{align}
Due to this time-translation symmetry, the block diagonal elements of the $S_{12}C_{12}$ composite evolve fully independently of the off-diagonal elements. Further, these block-diagonal elements can be expressed as the tensor-products of block-diagonal elements arising from the system $S_{12}$ and the catalyst $C_{12}$ corresponding to the energy eigenbases of their own Hamiltonians. For an initial state $\rho_{S_{12}}$ block-diagonal in energy bases, the block-diagonal part of the catalyst state only participates during the transformation. Therefore, a catalyst in a state block-diagonal in its energy eigenbases is enough.

\subsubsection{Avoiding negative $\alpha$}
The second laws for the transformations between system states, that are block-diagonal in energy, are based on the conditions \eqref{eq:Cat2ndLaws} for $\alpha \in [-\infty, \infty]$. However, we can get rid of the negative $\alpha$ in the conditions by borrowing an ancillary system in a pure state. The only condition is that, after the transformation, we return it with good fidelity. Even a two-qubit system in a pure state is enough to lift all the conditions involving negative $\alpha$.

\begin{thm} \label{thm:Cat2ndLawsPosiAlpha}
	Consider two states $\rho_{S_{12}}$ and  $\sigma_{S_{12}}$ that are block-diagonal in the energy eigenbases, with the associated Hamiltonian $H_{S_{12}}=H_{S_1} + H_{S_2}$. Additionally, we are allowed to borrow a two-qubit system $A_{12}$ with a trivial Hamiltonian and in a pure state $\proj{0}_{A_1} \otimes \proj{0}_{A_2}$, and then return it with a good fidelity. Then a transformation $ (\rho_{S_{12}}, H_{S_{12}}) \rightarrow (\sigma_{S_{12}}, H_{S_{12}})$ is possible under semi-local thermal operation if, and only if,
	\begin{align}\label{eq:Cat2ndLawsPosiAlpha}
		S_\alpha(\rho_{S_{12}}, \gamma_{S_1} \otimes \gamma_{S_2}) \geqslant S_\alpha(\sigma_{S_{12}}, \gamma_{S_1} \otimes \gamma_{S_2}), \ \ \forall \ \alpha \geqslant 0.
	\end{align}
\end{thm}
\begin{proof}
	Let us first assume that the transformation 
	\begin{align}
		\rho_{S_{12}} \otimes \proj{0}_{A_1} \otimes \proj{0}_{A_2} \rightarrow  \sigma_{S_{12}} \otimes \proj{0}_{A_1} \otimes \proj{0}_{A_2} 
	\end{align}
	is possible by means of a catalytic semi-local thermal operation. Then, using Theorem \ref{thm:2ndLaws} and noticing that the R\'enyi $\alpha$-relative entropies ($D_{\alpha}$) are finite for the state $\proj{0}_{A_1} \otimes \proj{0}_{A_2}$ only for $\alpha \geqslant 0$, we have  
	\begin{align}\label{eq:Cat2ndLawsDPosiAlpha}
		D_\alpha(\rho_{S_{12}} \otimes \proj{0}_{A_1} \otimes \proj{0}_{A_2} \parallel \gamma_{S_1} \otimes \gamma_{S_2} \otimes \gamma_{A_1} \otimes \gamma_{A_2}) \geqslant D_\alpha(\sigma_{S_{12}} \otimes \proj{0}_{A_1} \otimes \proj{0}_{A_2} \parallel \gamma_{S_1} \otimes \gamma_{S_2} \otimes \gamma_{A_1} \otimes \gamma_{A_2}), \ \ \forall \alpha \geqslant 0.
	\end{align}
	Moreover, $D_\alpha(\rho_{S_{12}} \otimes \proj{0}_{A_1} \otimes \proj{0}_{A_2} \parallel \gamma_{S_1} \otimes \gamma_{S_2} \otimes \gamma_{A_1} \otimes \gamma_{A_2})= D_\alpha(\rho_{S_{12}} \parallel \gamma_{S_1} \otimes \gamma_{S_2}) + D_\alpha(\proj{0}_{A_1} \otimes \proj{0}_{A_2} \parallel \gamma_{A_1} \otimes \gamma_{A_2})$. Thus the conditions \eqref{eq:Cat2ndLawsDPosiAlpha}, in turn, imply the conditions \eqref{eq:Cat2ndLawsPosiAlpha}. 
	
	Conversely, let us consider that the conditions \eqref{eq:Cat2ndLawsPosiAlpha} (as well as the conditions \eqref{eq:Cat2ndLawsDPosiAlpha}) are satisfied. As we have indicated earlier, the $D_{\alpha}$s become infinite with $\proj{0}_{A_1} \otimes \proj{0}_{A_2}$ for $\alpha <0$. However, we may allow that the final state of the ancillary system $A_{12}$ is returned in the full-rank state but arbitrarily close to the original state. Then the left-hand side of \eqref{eq:Cat2ndLawsDPosiAlpha} remains infinite but the right-hand side becomes finite. Thus, we are led to  
	\begin{align}
		D_\alpha(\rho_{S_{12}} \parallel \gamma_{S_1} \otimes \gamma_{S_2}) \geqslant D_\alpha(\sigma_{S_{12}} \parallel \gamma_{S_1} \otimes \gamma_{S_2}), \ \ \forall \alpha \in \mathbb{R}.
	\end{align}
	Now by Theorem \ref{thm:2ndLaws}, we say that the state $\rho_{S_{12}} \otimes \proj{0}_{A_1} \otimes \proj{0}_{A_2} $ can be transformed arbitrarily close to the state $\sigma_{S_{12}} \otimes \proj{0}_{A_1} \otimes \proj{0}_{A_2} $.
\end{proof}

\subsection{State transformation with time dependent Hamiltonians \label{sec:StateTransTimeDepHamil}}
So far we have restricted ourselves to the cases where the system Hamiltonian remains unchanged before and after the transformations. However, in real situations this restriction is not often respected. To include all these scenarios, we consider the cases where such changes in Hamiltonians are allowed.  

Consider a situation where the non-interacting Hamiltonian $H_{S_{12}}=H_{S_1}+H_{S_2}$ of the system $S_{12}$ changes to $H_{S_{12}}^\prime=H_{S_1}^\prime+H_{S_2}^\prime$, along with the state transformation $\rho_{12} \rightarrow \sigma_{12}^\prime$. Here the $(^\prime)$ indicates the state with modified Hamiltonian. Such a change in Hamiltonian often happens due to some time dependencies of the joint Hamiltonian. Then the second laws that incorporate such situations are given in the theorem below.
\begin{thm}[Second law for block-diagonal states] \label{thm:2ndLawsAll}
	Under a catalytic semi-local thermal operation, a transformation $\left( \rho_{S_{12}}, H_{S_{12}} \right) \longrightarrow \left( \sigma_{S_{12}}^\prime, H_{S_{12}}^\prime \right)$ that leads to changes both in system states and the non-interacting system Hamiltonians is possible if, and only if,
	\begin{align}
		S_{\alpha}\left(\rho_{S_{12}}, \gamma_{S_{1}} \otimes \gamma_{S_{2}} \right) \geqslant S_{\alpha}\left(\sigma_{S_{12}}^\prime, \gamma_{S_{1}}^\prime \otimes \gamma_{S_{2}}^\prime \right), \ \ \forall \alpha \geqslant0,
	\end{align}
	where $\gamma_{S_i}=\frac{e^{-\beta_i H_{S_i}}}{\tr [e^{-\beta_i H_{S_i}}]}$ and $\gamma_{S_i}^\prime=\frac{e^{-\beta_i H_{S_i}^\prime}}{\tr [e^{-\beta_i H_{S_i}^\prime}]}$, for $i=1,2$.
\end{thm}

\begin{proof}
	Let us assume the total Hamiltonian of the universe is time-independent. Any change happening in the system Hamiltonian can then be understood due to the presence of a clock system $X$. The joint Hamiltonian of the system and clock system is given by 
	\begin{align}
		\sum_{t=t_{i}}^{t_{f}} H(t) \otimes \proj{t}_X,
	\end{align}
	where $\ket{t}_X$ are some orthonormal basis. Then a change in Hamiltonian $H(t_i) \rightarrow H(t_f)$, along with a transformation in the system state $\rho \rightarrow \sigma$, is equivalent to the change in a joint clock-system state as 
	\begin{align}
		\rho(t_i) \otimes \proj{t_{i}}_X \longrightarrow \sigma(t_f) \otimes \proj{t_f}_X. 
	\end{align}
	We exploit this argument to change the joint non-interacting Hamiltonian of the system, so that
	\begin{align}
		H_{S_{12}}=H_{S_1}+H_{S_2}  \longrightarrow H_{S_{12}}^\prime=H_{S_1}^\prime+H_{S_2}^\prime.
	\end{align}
	We consider a bipartite clock system $X_{12}$ that is responsible for the change in Hamiltonian of the system $S_{12}$. Since the transition does not depend on the intermediate times, we simply assume the initial state to be $\ket{t_i}_{X_{1/2}}=\ket{0}_{X_{1/2}}$ and the final state to be $\ket{t_f}_{X_{1/2}}=\ket{1}_{X_{1/2}}$. Therefore, the time-independent joint Hamiltonian of the system and the clock can be written as 
	\begin{align}
		H_{S_{12}X_{12}}=\left(H_{S_{1}} \otimes \proj{0}_{X_{1}} + H_{S_{1}}^\prime \otimes \proj{1}_{X_{1}} \right) + \left(H_{S_{2}} \otimes \proj{0}_{X_{2}} + H_{S_{2}}^\prime \otimes \proj{1}_{X_{2}} \right)= H_{S_1X_1} +  H_{S_2X_2}.
	\end{align}
	Here $X_{12}$ plays the role of a switch and it controls the Hamiltonian on the system by choosing its state $\proj{ij}_{X_{12}}=\proj{i}_{X_1} \otimes \proj{j}_{X_2}$. For examples, when the switch $X_{12}$ is in the state $\proj{00}_{X_{12}}$, it implements the Hamiltonian $H_{S_{12}}$ on the system $S_{12}$. On the other hand, when the switch is in the state $\proj{11}_{X_{12}}$, it switches the system Hamiltonian to $H_{S_{12}}^\prime$. Now consider a catalytic semi-local thermal operation that leads to the transformation 
	\begin{align}
		\left(\rho_{S_{12}}, H_{S_{12}}  \right) \longrightarrow \left(\sigma_{S_{12}}^\prime, H_{S_{12}}^\prime  \right). 
	\end{align}
	This is equivalent to the transformation between the joint system-switch states block-diagonal in energy, and with the fixed Hamiltonian $H_{S_{12}X_{12}}$,
	\begin{align}
		\left(\rho_{S_{12}}\otimes \proj{00}_{X_{12}}, H_{S_{12}X_{12}}  \right) \longrightarrow \left(\sigma_{S_{12}}^\prime \otimes \proj{11}_{X_{12}}, H_{S_{12}X_{12}}  \right) 
	\end{align}
	under a catalytic semi-local thermal operation applied jointly on the system and the clock. Now following the Theorem \ref{thm:Cat2ndLawsPosiAlpha}, we cast the necessary and sufficient conditions for the transformation, as
	\begin{align}
		S_{\alpha}\left(\rho_{S_{12}}\otimes \proj{00}_{X_{12}}, \gamma_{S_{1}X_{1}} \otimes \gamma_{S_{2}X_{2}} \right) 
		\geqslant S_{\alpha}\left(\sigma_{S_{12}}^\prime \otimes \proj{11}_{X_{12}}, \gamma_{S_{1}X_{1}} \otimes \gamma_{S_{2}X_{2}} \right), \ \ \forall \alpha\geqslant 0,
	\end{align}
	where $\gamma_{S_{i}X_{i}}=\frac{e^{-\beta_i H_{S_iX_i}}}{\tr [e^{-\beta_i H_{S_iX_i}}]}$, for $i=1,2$, are the thermal states correspond to the Hamiltonians $H_{S_iX_i}$ and inverse temperatures $\beta_i$. Moreover, we notice that 
	\begin{align}
		S_{\alpha}\left(\rho_{S_{12}}\otimes \proj{00}_{X_{12}}, \gamma_{S_{1}X_{1}} \otimes \gamma_{S_{2}X_{2}} \right) -  S_{\alpha}\left(\sigma_{S_{12}}^\prime \otimes \proj{11}_{X_{12}},  \gamma_{S_{1}X_{1}} \otimes \gamma_{S_{2}X_{2}} \right) = S_{\alpha}\left(\rho_{S_{12}}, \gamma_{S_{1}} \otimes \gamma_{S_{2}} \right).
	\end{align} 
	As a result, the necessary and sufficient conditions for the transformation reduce to
	\begin{align}
		S_{\alpha}\left(\rho_{S_{12}}, \gamma_{S_{1}} \otimes \gamma_{S_{2}} \right) \geqslant S_{\alpha}\left(\sigma_{S_{12}}^\prime, \gamma_{S_{1}}^\prime \otimes \gamma_{S_{2}}^\prime \right), \ \ \forall \alpha  \geqslant 0.
	\end{align}
\end{proof}

\section{Free-entropy distance, thermodynamic work and fundamentally irreversibility \label{sec:fed}}
In this section, we relate $\alpha$-free-entropies, introduced in the previous section, with thermodynamic works. A formal connection and equivalence between work and purity have been established in \cite{Horodecki13, Brandao15}. Here we briefly outline the approach presented in \cite{Horodecki13}. A thermalization or work extraction process leads an arbitrary system-bath state to a more mixed (or less pure) system-bath state for a given block with fixed total energy. These processes are nothing but randomization (noisy) processes and extensively studied in context purity resource theory \cite{Horodecki13a}. However, there is a subtlety we encounter here, compared to purity resource theory. Nevertheless, one may claim that the thermodynamics is nothing but a purity resource theory constrained by the temperature of the bath and the Hamiltonian of the system \cite{Brandao13, Horodecki13, Brandao15}.

Consider a bipartite system $S_{12}$, with initial non-interacting Hamiltonian $H_{S_{12}}=H_{S_1}+H_{S_2}$, in an initial state $\rho_{S_{12}}$. After a catalytic semi-local thermal operation, the state and the Hamiltonian for the system are changed to $\sigma_{S_{12}}^\prime$ and  $H_{S_{12}}^\prime=H_{S_1}^\prime + H_{S_2}^\prime$ respectively, i.e.,
\begin{align}\label{eq:GenTrans}
	\left(\rho_{S_{12}}, H_{S_{12}} \right) \longrightarrow \left( \sigma_{S_{12}}^\prime, H_{S_{12}}^\prime  \right). 
\end{align}
For this transformation we may need a catalyst $C_{12}$, with the Hamiltonian $H_{C_{12}}=H_{C_1}+ H_{C_2}$. However, for simplicity, we consider the catalyst as a part of the system.

Our aim is to exploit this transformation to extract free-entropy and thermodynamic work. For that, we also introduce a battery that stores or expends work. We may think that the bipartite battery $S_{W_{12}}$ is composed of sub-systems $S_{W_1}$ and $S_{W_2}$ with the Hamiltonian $H_{S_{W_{12}}}=H_{S_{W_1}}+H_{S_{W_2}}$, where the two-level Hamiltonians are $H_{S_{W_1}}=W_1\proj{W_1}$ and $H_{S_{W_2}}=W_2\proj{W_2}$. The battery sub-system $S_{W_1}$ ($S_{W_2}$) is semi-locally interacting with the bath $B_1$ ($B_2$). Note, these two battery sub-systems can in principle exchange energy, i.e., work, as this operation is allowed by the catalytic semi-local thermal processes. When the battery sub-systems are thermalized to the temperatures of $B_1$ and $B_2$, the corresponding semi-Gibbs state becomes $\gamma_{S_{W_1}} \otimes \gamma_{S_{W_2}}$, where
\begin{align}
	\gamma_{S_{W_i}}=\frac{1}{1+e^{-\beta_i W_i}} \left(\proj{0} + e^{-\beta_i W_i} \proj{1}  \right),
\end{align}
with $i=1,2$, and $\beta_i$ is the inverse temperature of the bath $B_i$.
Now, let us consider the transformation of the system and the battery together
\begin{align}
	\left( \rho_{S_{12}} \otimes \proj{00}_{S_{W_{12}}}, H_{S_{12}}+ H_{S_{W_{12}}} \right) \longrightarrow \left( \sigma_{S_{12}}^\prime \otimes \proj{00}_{S_{W_{12}}}, H_{S_{12}}^\prime + H_{S_{W_{12}}} \right),
\end{align}
where we denote $\proj{W_1W_2}_{S_{W_{12}}}=\proj{W_1}_{S_{W_{1}}} \otimes \proj{W_2}_{S_{W_{2}}}$. Then second laws, i.e., Theorem \ref{thm:2ndLawsAll}, ensure the conditions 
\begin{align}
	S_\alpha(\rho_{S_{12}} \otimes \proj{00}_{S_{W_{12}}}, \ \gamma_{S_1} \otimes \gamma_{S_2} \otimes \gamma_{S_{W_1}} \otimes \gamma_{S_{W_2}}) \geqslant S_\alpha(\sigma_{S_{12}}^\prime \otimes \proj{W_1W_2}_{S_{W_{12}}}, \ \gamma_{S_1}^\prime \otimes \gamma_{S_2}^\prime \otimes \gamma_{S_{W_1}} \otimes \gamma_{S_{W_2}})
\end{align}
to be satisfied for all $\alpha \geqslant 0$. Given the initial, the final and the thermal states of the battery, we can derive the bound on the $\alpha$-free-entropy stored in the battery, as
\begin{align}
	S_\alpha(\rho_{S_{12}},\gamma_{S_1} \otimes \gamma_{S_2})-S_\alpha(\sigma_{S_{12}}^\prime,\gamma_{S_1}^\prime \otimes \gamma_{S_2}^\prime)\geqslant \beta_1 W_1 + \beta_2 W_2,
\end{align}
where $W_1$ ($W_2$) is amount of work stored in the battery sub-system via the transformation in its state $\proj{0}_{W_1} \to \proj{W_1}_{W_1}$ ($\proj{0}_{W_2} \to \proj{W_2}_{W_2}$). Due to total energy and total-weighted energy conservation of the global process on the system-battery-baths composite, the battery $S_{W_{12}}$ can only increase its energy. This implies 
\begin{align}
	W_1+W_2 \geqslant 0.
\end{align}
However, note that the quantities $W_1$ and $W_2$ depend on the cSLTO that executes the transformation. Therefore, it is important to quantify the guaranteed amount of free-entropy or the works involved in a state transformation irrespective to the operations that execute it. This quantification is done in terms of the free-entropy distance, given below. 
\begin{thm}[Free-entropy distance]\label{thm:fed}
	For a catalytic semi-local thermal operation, leading to a transition $\left(\rho_{S_{12}}, H_{S_{12}} \right) \longrightarrow \left( \sigma_{S_{12}}^\prime, H_{S_{12}}^\prime  \right)$, the free-entropy distance between the initial and final block-diagonal states is given by
	\begin{align}
		S_d \left(\rho_{12} \rightarrow \sigma_{12}^\prime \right)= \inf_{\alpha \geqslant 0}\left[S_\alpha(\rho_{S_{12}},\gamma_{S_1} \otimes \gamma_{S_2}) -  S_\alpha(\sigma_{S_{12}}^\prime,\gamma_{S_1}^\prime \otimes \gamma_{S_2}^\prime) \right]= \beta_1 W_1 + \beta_2 W_2 \geqslant 0. \label{eq:free-ent-dist}
	\end{align}
\end{thm}
Note, the guaranteed extracted work is $W^{ext}=W_1 + W_2$. The Theorem \ref{thm:fed} leads to several interesting results. Now in terms of the free-entropy distance, we can quantify these quantities, as in the following. 

\begin{cor}[Extractable free-entropy and free-entropy cost]\label{Coro:fed}
	For a transformation between the block-diagonal states, $\left(\rho_{S_{12}}, H_{S_{12}} \right) \longrightarrow \left( \sigma_{S_{12}}^\prime, H_{S_{12}}^\prime  \right)$, under catalytic semi-local thermal operations, the extractable free-entropy $S_{ext}$, and the free-entropy cost $S_{cost}$ for the reverse the process, are given by
	\begin{align}
		&S_{ext} \left(\rho_{12} \rightarrow \sigma_{12}^\prime \right)=\inf_{\alpha \geqslant 0}\left[S_\alpha(\rho_{S_{12}},\gamma_{S_1} \otimes \gamma_{S_2}) -  S_\alpha(\sigma_{S_{12}}^\prime,\gamma_{S_1}^\prime \otimes \gamma_{S_2}^\prime) \right]=S_d \left(\rho_{12} \rightarrow \sigma_{12}^\prime \right), \nonumber \\
		&S_{cost} \left(\sigma_{12}^\prime \rightarrow \rho_{12}\right)=-\sup_{\alpha \geqslant 0}\left[S_\alpha(\rho_{S_{12}},\gamma_{S_1} \otimes \gamma_{S_2}) -  S_\alpha(\sigma_{S_{12}}^\prime,\gamma_{S_1}^\prime \otimes \gamma_{S_2}^\prime) \right].
	\end{align}
\end{cor}

It is clear that the free-entropy that we can extract from the process is in general lower than the free-entropy to be expended to reverse the process. To see this, let us consider the transformation \eqref{eq:GenTrans} and assume $S_{ext} (\rho_{S_{12}} \to \sigma_{S_{12}}^\prime) >0$. Then, 
\begin{align}
	S_{ext} (\rho_{S_{12}} \to \sigma_{S_{12}}^\prime) \leqslant - \ S_{cost} (\sigma_{S_{12}}^\prime \to \rho_{S_{12}}).
\end{align}
Therefore, thermodynamics in this regime is fundamentally irreversible, analogous to the cases with single bath \cite{Horodecki13}. With the help of free-entropy distance, we are also able to compute the distillable free-entropy and free-entropy of formation for a state. 

In the situation where the $\sigma_{S_{12}}=\gamma_{S_1} \otimes \gamma_{S_2}$ and the system Hamiltonian does not change, the  Corollary \ref{Coro:fed} leads us to quantify distillable free-entropy $S_{dist}$ for the process  $ \rho_{S_{12}} \rightarrow \gamma_{S_1} \otimes \gamma_{S_2}$, and the free-entropy of formation $S_{form}$ for the process $ \gamma_{S_1} \otimes \gamma_{S_2} \rightarrow  \rho_{S_{12}}$, as
\begin{align}
	&S_{dist}(\rho_{S_{12}})=D_0\left(\rho_{S_{12}} \parallel \gamma_{S_1} \otimes \gamma_{S_2} \right) = -\log \tr \left( \Pi_{\rho_{S_{12}}} \gamma_{S_1} \otimes \gamma_{S_2} \right), \\
	&S_{form}(\rho_{S_{12}})=D_\infty\left(\rho_{S_{12}} \parallel \gamma_{S_1} \otimes \gamma_{S_2} \right) = \log \min \left\{\lambda: \rho_{S_{12}} \leqslant \lambda (\gamma_{S_1} \otimes \gamma_{S_2}) \right\}.
\end{align}

\section{Superposition and free-entropy locking}
The cSLTOs are time translation symmetric with respect to the system Hamiltonian $H_{S_{12}}$. Consider a quantum state $\rho_{S_{12}}$ that has superposition in the eigenbases of $H^{\beta_1 \beta_2}_{S_{12}}$. Then the state can be expressed as
\begin{align}
	\rho_{S_{12}}= \rho_{S_{12}}^{d} + \rho_{S_{12}}^{o}, 
\end{align}
where $\rho_{S_{12}}^{d}$ is the block-diagonal part of the state and $\rho_{S_{12}}^{o}$ is the off-diagonal part of the state when it is expressed in the energy eigenbases of the system. The off-diagonal part $\rho_{S_{12}}^{o}$ evolve independently of the block-diagonal part $\rho_{S_{12}}^{d}$ under cSLTOs, and the cSLTOs can only access the block-diagonal elements for free-entropy extraction. That is why the only accessible free-entropy in the one-shot finite-size regime is the one corresponding to the dephased state of the original one.

Therefore, there is a free-entropy locking in the presence of such superposition as the free-entropy stored in the superposition cannot be accessed. Note, in presence of quantum correlation, e.g., entanglement, such superposition is inevitably present in the state, and there would be free-entropy locking.  However, in the asymptotic regime, where the number of particles in the system becomes considerably large, this free-entropy can be unlocked and fully accessed via cSLTOs (see below). As a result, one can fully extract free-entropy from the quantum correlations present in the system in the asymptotic regime.

\section{Second laws for the non-block-diagonal states \label{secApp:2ndLawsAssymetry}}
We have already mentioned that the second laws derived in the Theorem~\ref{thm:2ndLawsAll} provide the necessary and sufficient conditions for the transformation between states that are block-diagonal in the eigenbases of the Hamiltonian $H_{S_{12}}$. For the states that are not block-diagonal, the second laws become only the necessary conditions. These necessary conditions can be further supplemented by a monotonic measure of the time-translation asymmetry. For a quantum state $\rho_{S_{12}}$ with the Hamiltonian $H_{S_{12}}$, the asymmetry is quantified as 
\begin{align}
	A_{\alpha}(\rho_{S_{12}}, H_{S_{12}})=D^q_{\alpha}\left(\rho_{S_{12}} \parallel P_{S_{12}} (\rho_{S_{12}})  \right)
\end{align}
where the $P_{S_{12}} (\rho_{S_{12}})$ is the dephased state in the energy eignebases, given in Eq. \eqref{eqApp:DephasedStateWeightedHamil}, and the quantum R\'enyi $\alpha$-relative entropy is defined as
\begin{align}
	D^q_{\alpha}(\rho \parallel \sigma)=
	\begin{cases}
		\frac{1}{\alpha - 1} \log \tr(\rho^\alpha \sigma^{1-\alpha}),  \ & \alpha \in [0, \ 1) \\
		\frac{1}{\alpha - 1} \log \tr \left( [\sigma^{\frac{1-\alpha}{2 \alpha}} \ \rho \ \sigma^{\frac{1-\alpha}{2 \alpha}}]^\alpha \right),   \ & \alpha >0.
	\end{cases}
\end{align}
Note, in the limit $\alpha \to 1$, the quantum $\alpha$-relative entropy converges to the well known von Neumann relative entropy $D_1^q(\rho \parallel \sigma)=\tr (\rho \log \rho - \rho \log \sigma)$. 

Similar to the state transformation in the presence of a single bath \cite{Lostaglio15a, Lostaglio15}, we also supplement the necessary condition for state transformation in presence of two baths via cSLTOs, in addition to the second laws for the diagonal states. Consider a transformation $(\rho_{S_{12}}, H_{S_{12}}) \to (\sigma_{S_{12}}^\prime, H_{S_{12}}^\prime)$ via a cSLTO, where the states may have superposition in the energy eigenbases. Then, the necessary conditions for the transformation are
\begin{align}
	A_{\alpha}(\rho_{S_{12}}, H_{S_{12}}) \geqslant A_{\alpha}(\sigma_{S_{12}}^\prime, H_{S_{12}}^\prime), \ \ \forall \alpha \geqslant 0,
\end{align}
where the amount quantum asymmetry monotonically decreases under cSLTOs.

\section{Asymptotic equipartition \label{secApp:asymptotic}}
The framework developed for quantum heat engines working in the one-shot finite-size regime can reproduce the know results of thermodynamics for the engines operating in the asymptotic regime. It is well known that all the R\'enyi $\alpha$-relative entropies converge to von Neuman relative entropy in the asymptotic regime (or i.i.d. regime), that is
\begin{align}
	&\lim_{N \to \infty} \frac{1}{N} D_{\alpha}(\rho^{\otimes N} \parallel \sigma^{\otimes N})=D_1(\rho \parallel \sigma)=\tr(\rho \log \rho - \rho \log \sigma), \nonumber
\end{align}
for all $\alpha$. For that reason, the $\alpha$-free-entropies reduce to the Helmholtz free-entropy for all $\alpha$, as 
\begin{align}
	\lim_{N \to \infty} \frac{1}{N} S_{\alpha} \left(\rho_{S_{12}}^{\otimes N}, \ (\gamma_{S_1} \otimes \gamma_{S_2})^{\otimes N}\right)=S_1 (\rho_{S_{12}}, \ \gamma_{S_1} \otimes \gamma_{S_2}).
\end{align}
Therefore, in the asymptotic regime, there is only one free-entropy and that is the Helmholtz free-entropy. Consider a transformation in the asymptotic regime via a cSLTO, where the individual system transformations as 
\begin{align}\label{eqmt:asymt_trans}
	(\rho_{S_{12}}, H_{S_{12}}) \to (\sigma_{S_{12}}^\prime, H_{S_{12}}^\prime). 
\end{align}
The many second laws in the one-shot finite-size regime converge to a single second law providing necessary and sufficient condition for the engine transformation in the asymptotic regime, that is
\begin{align}
	S_1 (\rho_{S_{12}}, \ \gamma_{S_1} \otimes \gamma_{S_2}) \geqslant S_1 (\sigma_{S_{12}}^\prime, \ \gamma_{S_1}^\prime \otimes \gamma_{S_2}^\prime).
\end{align}
The amount of extractable free-entropy or the free-entropy distance per copy of the system, given the initial and the final states, is
\begin{align}
	S_d(\rho_{S_{12}} \to \sigma_{S_{12}}^\prime)&=S_1 (\rho_{S_{12}}, \ \gamma_{S_1} \otimes \gamma_{S_2}) - S_1 (\sigma_{S_{12}}^\prime, \ \gamma_{S_1}^\prime \otimes \gamma_{S_2}^\prime) \nonumber \\
	&= \Delta S - \beta_1 \Delta E_{S_1} - \beta_2 \Delta E_{S_2}, 
\end{align}
where $\Delta S=S(\rho_{S_{12}})-S(\sigma_{S_{12}}^\prime)$ the change in the von Neumann entropy. $\Delta E_{S_{1}}=\tr \rho_{S_{12}}H_{S_1} - \tr \sigma_{S_{12}}^\prime H_{S_1}^\prime $ is the change in average energy of the subsystem $S_1$, and similarly, the $\Delta E_{S_{2}}$ is for $S_{2}$.
Note, for the reverse transformation $(\sigma_{S_{12}}^\prime, H_{S_{12}}^\prime) \to (\rho_{S_{12}}, H_{S_{12}}) $, the free-entropy cost per copy of the system is exactly equal to the extractable free-entropy in the forward process, i.e., 
\begin{align}
	S_d(\rho_{S_{12}} \to \sigma_{S_{12}}^\prime)=  S_d(\rho_{S_{12}} \leftarrow \sigma_{S_{12}}^\prime),
\end{align}
As a consequence, the thermodynamic reversibility is recovered in the asymptotic regime.

We recall that there is a free-entropy locking in the one-shot finite-size regime for the quantum states that have quantum superposition in the energy eigenbases, i.e., for $[\rho_{S_{12}}, H_{S_{12}}]\neq 0$. However, in the asymptotic limit, where the number of systems $N \to \infty$, all states become symmetric with respect to the Hamiltonian on average, and this is because of the fact that 
\begin{align}
	\lim_{N \to \infty} \frac{1}{N} \left[ \rho_{S_{12}}^{\otimes N},  \ \sum_{x=0}^{N-1} \mathbb{I}^{\otimes x} \otimes H_{S_{12}} \otimes \mathbb{I}^{\otimes (N-x-1)} \right]=0, \ \ \forall \ \rho_{S_{12}}. \nonumber
\end{align}
So, an arbitrary state can be written in the block-diagonal form on average in the asymptotic regime. It means that a state which is non-block-diagonal in single-copy level becomes block-diagonal (on average) in the energy eigenbases in the asymptotic regime. Consequently, the locked free-entropy due to the presence of quantum superpositions and correlations in a state can be accessed and extracted.

Let us now turn to the heat engines and show how the framework presented above reproduces the traditional thermodynamics and its laws. One may recover the statements of the second law in terms of heat, where the heat $Q$ is defined as $Q=\Delta E - W$. Here $\Delta E$ is the change in the internal energy in the system and $W$ is the work done by the system. Consider an engine transformation where bipartite system $S_{12}$ is semi-locally interacting with the baths $B_1$ and $B_2$ via semi-local thermal operation. The global operation respect strict conservation of total weighted-energy, i.e., $\beta_1 \Delta E_1 + \beta_2 \Delta E_2=0$, where $\Delta E_1$ and $\Delta E_2$ is the change in energies of the $B_1S_1$ and $B_2S_2$ composites respectively. Then, the total weighted-energy conservation and Eq.~\eqref{eq:free-ent-dist} together lead to
\begin{align}\label{eq:AllStatementsWithHeat}
	\beta_1 Q_1 + \beta_2 Q_2 \leqslant 0,
\end{align}
where $Q_{1/2}=\Delta E_{1/2} - W_{1/2}$. The expression above is the Clausius inequality and mathematically captures all the statements of the second law in terms of heat. For instance, in the asymptotic regime where the above definition of heat is applicable, we recover the traditional form of Carnot efficiency as $\eta=\frac{W_{\infty}}{Q_1}\leqslant \frac{Q_1 + Q_2}{Q_1} \leqslant 1- \frac{\beta_1}{\beta_2}$, where the extracted work is given by $W_{\infty} \leqslant Q_1 + Q_2$ and equality only holds for the reversible engine operations. 

In the asymptotic regime, the role of correlations in driving ``anomalous'' heat flow from a cold to a hot bath can also be understood. For example, consider the situation where engine operation leads to a transformation $(\rho_{S_{12}}, H_{S_{12}}) \to (\rho_{S_{1}} \otimes \rho_{S_2}, H_{S_{12}})$ which exploits the correlation between the subsystems $S_1$ and $S_2$. Here $\rho_{S_1}=\tr_{S_2}\rho_{S_{12}}$ and $\rho_{S_2}=\tr_{S_1}\rho_{S_{12}}$. As per the second law in the asymptotic regime, the transformation takes place spontaneously if and only if the free-entropy satisfy $S_1(\rho_{S_{12}}, \gamma_{S_1} \otimes \gamma_{S_2}) \geqslant S_1(\rho_{S_{1}} \otimes \rho_{S_2}, \gamma_{S_1} \otimes \gamma_{S_2})$. The extractable free-entropy from this process is
\begin{align}
	\mathcal{I}(S_1:S_2) \geqslant \beta_1 W_1 + \beta_2 W_2 \geqslant 0.
\end{align}
where $\mathcal{I}(S_1:S_2)=S(\rho_{S_1})+ S(\rho_{S_2}) - S(\rho_{S_{12}})$ is the mutual information quantifying the correlation present in the system $S_{12}$. The $W^{ext}=W_1 + W_2 \geqslant 0$ is the extractable work. This stored work in correlation is responsible for the spontaneous ``anomalous'' heat flow from the cold to the hot bath which is nothing but a process happens in case of refrigeration.

\end{document}